%% file: main.tex
\documentclass[11pt,a4paper]{article}
\input{header}

\usepackage{graphics}
\usepackage{graphicx}

\usepackage{subcaption}
\usepackage{float}
\usepackage{verbatim}
\usepackage{url}

\begin{document}

\title {Finding large $k$-colorable induced subgraphs\\ in (bull, chair)-free and (bull,E)-free graphs}
\date{}

\author[1]{Nadzieja Hodur}
\author[1]{Monika Pil\'sniak}
\author[1]{\\Magdalena Prorok}
\author[2,3]{Pawe{\l} Rz\k{a}\.zewski\footnote{This research was funded in whole or in part by National Science Centre grant number 2024/54/E/ST6/00094. For the purpose of Open Access, the author has applied a CC-BY public copyright licence to any
Author Accepted Manuscript (AAM) version arising from this submission.}}

\affil[1]{AGH University of Krakow, al. Mickiewicza 30, 30-059 Krak\'ow, Poland}
\affil[2]{Warsaw University of Technology, Poland}
\affil[3]{University of Warsaw, Poland}

\date{\today}

\maketitle

\begin{abstract}
    We study the \textsc{Max Partial $k$-Coloring} problem, where we are given a vertex-weighted graph, and we ask for a maximum-weight induced subgraph that admits a proper $k$-coloring.
    For $k=1$ this problem coincides with \textsc{Maximum Weight Independent Set}, and for $k=2$ the problem is equivalent (by complementation) to \textsc{Minimum Odd Cycle Transversal}.
    Furthermore, it generalizes \textsc{$k$-Coloring}.

    We show that \textsc{Max Partial $k$-Coloring} on $n$-vertex instances with clique number $\omega$ can be solved in time
    \begin{itemize}
        \item $n^{\mathcal{O}(k\omega)}$ if the input graph excludes the bull and the chair as an induced subgraph,
        \item $n^{\mathcal{O}(k\omega \log n)}$ if the input graph excludes the bull and \textsf{E} as an induced subgraph.
    \end{itemize}    
    This implies that \textsc{$k$-Coloring} can be solved in polynomial time in the former class, and in quasipolynomial-time in the latter one.
\end{abstract}


\section{Introduction}
Graph coloring is undoubtedly one of the most studied notions in graph theory, both from the structural and from the algorithmic point of view. For a positive integer $k$, in the $k$-\textsc{Coloring} problem we are given a graph~$G$, and we ask whether $G$ admit a proper $k$-coloring of $G$, i.e., an assignment of labels from $\{1,\ldots,k\}$ to the vertices of $G$ so that adjacent vertices receive distinct labels.
It is well-known that for each $k \geq 3$, the \hbox{$k$-\textsc{Coloring}} problem is \textsf{NP}-hard~\cite{DBLP:journals/sigact/Stockmeyer73}.
However, such a hardness result is typically just a start of further research questions:
What makes the problem hard?
Which instances are actually hard?
This motivates the study of the complexity of the problem for restricted graph classes, in hope to understand the boundary between tractable and intractable cases.

A natural way of obtaining such graph classes is by forbidding certain substructures.
For graphs $G,F$, we say that $G$ is \emph{$F$-free} is $F$ is not an induced subgraph of $G$.
In other words, we cannot obtain $F$ from~$G$ by deleting vertices.
For a family $\cF$ of graphs, we say that $G$ is \emph{$\cF$-free} if it $F$-free for every $F \in \cF$.
If $\cF = \{F_1,\ldots,F_k\}$, we usually write $(F_1,\ldots,F_k)$-free instead of $\{F_1,\ldots,F_k\}$-free.

Note that for each $\cF$, the family of $\cF$-free graphs is \emph{hereditary}, i.e., closed under vertex deletion.
Conversely, every hereditary family can be equivalently defined as $\cF$-free graphs, for some (possibly infinite)~$\cF$.

The study of $k$-\textsc{Coloring} in hereditary graph classes is an active and fruitful area of research.
It is known that for the case of $F$-free graphs, i.e., if we forbid just one induced subgraph,
the problem remains \textsf{NP}-hard for every $k \geq 3$ unless $F$ is a linear forest (i.e., a forest of paths)~\cite{DBLP:journals/siamcomp/Holyer81a,DBLP:journals/cpc/Emden-WeinertHK98,DBLP:journals/jal/LevenG83}.
Let us focus on the case that $F$ is connected, i.e., it is a path on $t$ vertices, denoted by $P_t$.
If $t \leq 5$, then $k$-\textsc{Coloring} is polynomial-time solvable for every $k$~\cite{DBLP:journals/algorithmica/HoangKLSS10}.
Furthermore, we know that if $t = 6$, then $k$-\textsc{Coloring} is polynomial-time solvable for $k \leq 4$~\cite{DBLP:journals/dam/RanderathS04,DBLP:journals/siamcomp/ChudnovskySZ24,DBLP:journals/siamcomp/ChudnovskySZ24a} and \textsf{NP}-hard otherwise~\cite{DBLP:journals/ejc/Huang16}.
It is also known that for $P_7$-free graphs, $k$-\textsc{Coloring} is polynomial-time solvable for $k = 3$~\cite{DBLP:journals/combinatorica/BonomoCMSSZ18} and \textsf{NP}-hard for $k \geq 4$~\cite{DBLP:journals/ejc/Huang16}.
Finally, for $P_8$-free graphs, $k$-\textsc{Coloring} is \textsf{NP}-hard for all $k \geq 4$~\cite{DBLP:journals/tcs/BroersmaGPS12}.

This leaves open the complexity of $3$-\textsc{Coloring} in $P_t$-free graphs for $k \geq 8$,
which is one of the notorious open problems in algorithmic graph theory.
Interestingly, it was shown that for every fixed $t$, the $3$-\textsc{Coloring} problem in $P_t$-free graphs can be solved in \emph{quasipolynomial time}~\cite{DBLP:conf/sosa/PilipczukPR21}. This is a strong indication that none of the remaining cases is \textsf{NP}-hard.

Let us remark that almost all mentioned algorithmic results work for the more general \textsc{List}-$k$-\textsc{Coloring} problem, where each vertex $v$ of the input graph $G$ is equipped with a \emph{list} $L(v) \subseteq \{1,\ldots,k\}$, and we ask for a proper coloring assigning to each vertex a color from its list.
The only exception is the case of $k=4$ and $t=6$. Indeed, \textsc{List}-$4$-\textsc{Coloring} is \textsf{NP}-hard in $P_6$-free graphs~\cite{DBLP:journals/iandc/GolovachPS14}.

There are also some results concerning (\textsc{List}-)$k$-\textsc{Coloring} of $\cF$-free graphs, when $\cF$ contains more graphs~\cite{DBLP:journals/tcs/Bonomo-Braberman21,DBLP:journals/dm/ChudnovskySZ20,DBLP:journals/jgt/ChudnovskyMSZ17,DBLP:journals/algorithmica/JelinekKMNP22,DBLP:journals/jcss/GaspersHP19,DBLP:journals/ipl/DabrowskiP18,DBLP:journals/endm/BrauseSHRVK15}; see also survey papers by Schiermeyer~\cite{DBLP:journals/gc/RanderathS04} and Golovach, Johnson, Paulusma, and Song~\cite{DBLP:journals/jgt/GolovachJPS17}.

Our main motivation was the recent paper of Hodur, Pilśniak, Prorok, and Schiermeyer~\cite{hodur20243colourabilitybullhfreegraphs}.
Among other results, they proved that $3$-\textsc{Coloring} can be solved in polynomial time in the class of $(\bull,\E)$-free, where $\bull$ and $\E$ are the graphs depicted in \cref{pic:forbidden}.\footnote{Actually, they showed a stronger result, as they characterized all minimal $(\bull,\E)$-free graphs that \emph{not} 3-colorable.}
The goal of this paper is to generalize their result to (\textsc{List}-)$k$-\textsc{Coloring}, for any $k \geq 3$.

\begin{figure}[t]
    \centering    
    \begin{subfigure}{0.3\textwidth}
    \centering
    \includegraphics[width=0.4\linewidth]{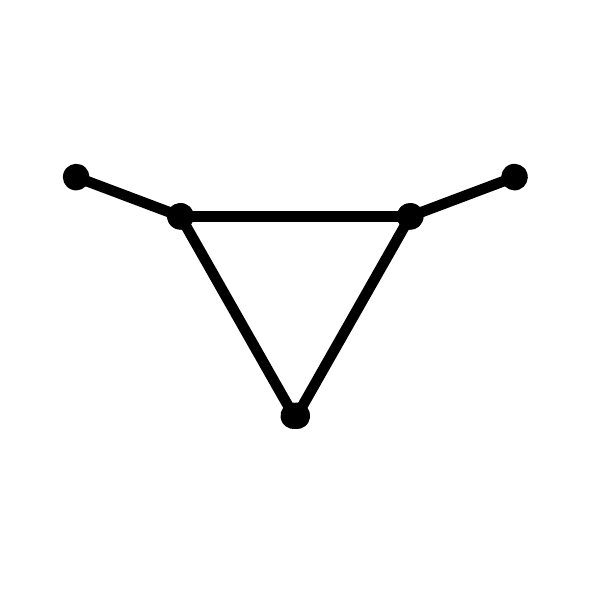} 
    \caption{\bull}
    \end{subfigure}
    \begin{subfigure}{0.3\textwidth}
    \centering
    \includegraphics[width=0.4\linewidth]{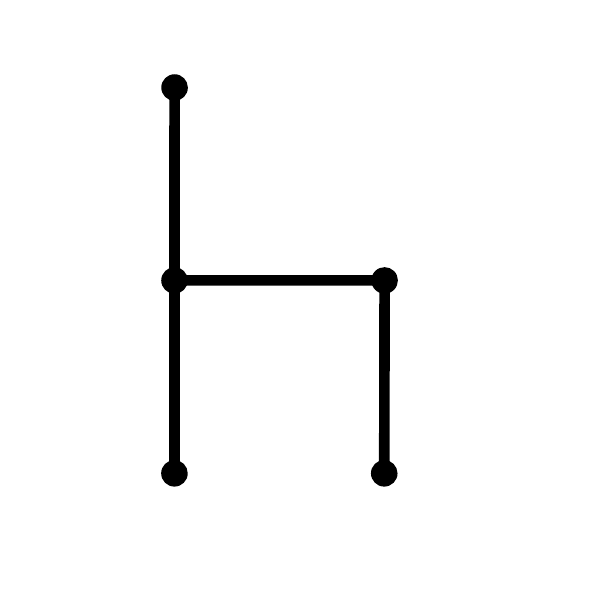}
    \caption{\chair}
    \end{subfigure}
    \begin{subfigure}{0.3\textwidth}
    \centering
    \includegraphics[width=0.4\linewidth]{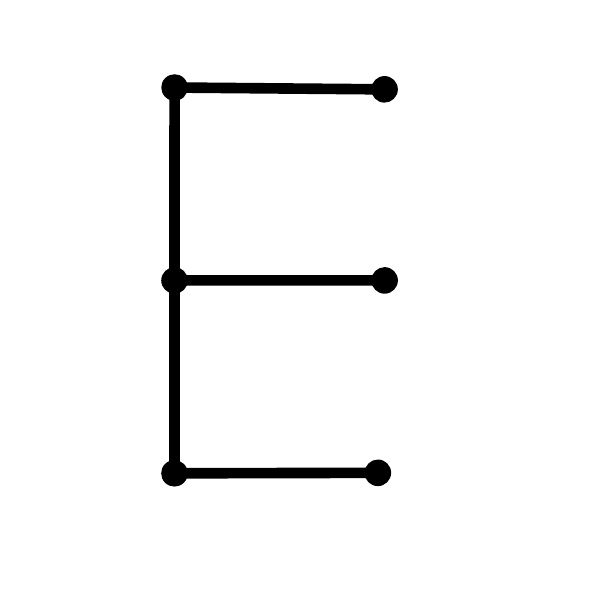}
    \caption{\E}    
    \end{subfigure}
    \caption{Considered forbidden subgraphs.}\label{pic:forbidden}
\end{figure}

In fact, we solve an even further generalization of $k$-\textsc{Coloring}, defined as follows.
\problemTask{\textsc{Max Partial $k$-Coloring}}
{a graph $G$, a revenue function $\rev : V(G) \times \{1,\ldots,k\} \to \mathbb{Q}_{\geq 0}$}
{a set $X$ and a proper $k$-coloring $c$ of $G[X]$, such that \[\sum_{v \in X} \rev(v,c(v))\] is maximum possible }

We can think of the value of $\rev(v,c)$ as the prize we get for coloring a vertex with color $c$.
We aim to find a partial coloring that maximizes the total prize.

Clearly this formalism captures $k$-\textsc{Coloring}: it is sufficient to set $\rev(v,c) = 1$ for every $v$ and $c$,
and check whether the maximum revenue is equal to the number of vertices.
If we additionally want to express lists $L : V(G) \to \{1,\ldots,k\}$, i.e., capture \textsc{List}-$k$-\textsc{Coloring}, we can set $\rev(v,c)=0$ if $c \notin L(v)$ and $\rev(v,c)=1$ otherwise, and again ask if there is a solution with total revenue $|V(G)|$.
However, \textsc{Max Partial $k$-Coloring} generalizes more problems.
If $k = 1$, then the problem is equivalent to \textsc{Max Weight Independent Set} problem.
If $k = 2$, then we ask for a maximum-weight induced bipartite subgraph, which is by complementation equivalent to \textsc{Odd Cycle Transversal}. Both these problems received a considerable attention from the algorithmic graph theory community~\cite{DBLP:conf/stoc/GartlandLMPPR24,DBLP:conf/focs/GartlandL20,DBLP:conf/sosa/PilipczukPR21,DBLP:conf/soda/LokshantovVV14,DBLP:journals/talg/GrzesikKPP22,DBLP:journals/siamcomp/AbrishamiCPRS24,DBLP:journals/algorithmica/DabrowskiFJPPR20,DBLP:journals/siamdm/ChudnovskyKPRS21,10.1145/3708544}.

We consider classes of $(\bull,\chair)$-free and $(\bull,\E)$-free graphs; we refer to \cref{pic:forbidden} for the definition of the forbidden induced subgraphs. Note that $\chair$ is an induced subgraph of $\E$ and thus  $(\bull,\chair)$-free graphs form a proper subclass of  $(\bull,\E)$-free graphs.
The following two theorems are the main result of our paper.

\begin{theorem}\label{thm:chair}
    For every $k \geq 1$, \textsc{Max Partial $k$-Coloring} on $(\bull,\chair)$-free instances with $n$ vertices and clique number $\omega$ can be solved in time $n^{\Oh(k \omega)}$.
\end{theorem}

\begin{theorem}\label{thm:E}
    For every $k \geq 1$, \textsc{Max Partial $k$-Coloring} on $(\bull,\E)$-free instances with $n$ vertices and clique number $\omega$ can be solved in time $n^{\Oh(k \omega  \log n)}$.
\end{theorem}

While the running time in \cref{thm:E} is not polynomial, but quasipolynomial in $n$, it still gives a strong evidence that the problem is not \textsf{NP}-hard. Indeed, in such a case all problems in \textsf{NP} can be solved in quasipolynomial time, which is unlikely according to our current understanding of complexity theory.

Interestingly, the algorithm in \cref{thm:chair,thm:E} is exactly the same, the only difference is the complexity analysis. We start with a careful analysis of the structure of $(\bull,\E)$-free graphs.
We observe that after exhaustively guessing a constant number of vertices and their color, we can decompose the input graph into parts that (1) are ``simpler'' and (2) the connections between the parts are ``well-structured.''
The first property allows us to call the algorithm recursively for each part, in order to obtain their corresponding partial solutions. Then the second property is used to combine these partial solutions into the solution of the input instance.
By ``simpler'' we typically mean that the clique number of a part is smaller than the clique number of the graph itself, but in one case (for $(\bull,\E)$-free graphs) ``simpler'' means just ``multiplicatively smaller.''
This explains the running time in \cref{thm:chair,thm:E}.
Let us remark that the idea of using the clique number to bound the complexity of an algorithm already appears in the literature~\cite{DBLP:journals/siamdm/ChudnovskyKPRS21,DBLP:journals/corr/abs-2412-14836}.

Notice that if we are just interested in solving \textsc{List}-$k$-\textsc{Coloring}, we can safely assume that the clique number of each instance is at most $k$: otherwise we can safely reject. Thus we immediately obtain the following corollaries.

\begin{restatable}{corollary}{corchair}
\label{cor:chair}
    For every $k \geq 3$, \textsc{List $k$-Coloring} on $n$-vertex $(\bull,\chair)$-free graphs can be solved in time $n^{\Oh(k^2)}$.
\end{restatable}

\begin{restatable}{corollary}{corE}
\label{cor:E}
    For every $k \geq 3$,  \textsc{List $k$-Coloring} on $n$-vertex $(\bull,\E)$-free graphs can be solved in time $n^{\Oh(k^2 \log n)}$.
\end{restatable}

Furthermore, using the win-win approach of Chudnovsky et al.~\cite{DBLP:journals/siamdm/ChudnovskyKPRS21}, we can show that for every $k$, the \textsc{Max Partial $k$-Coloring} problem can be solved in \emph{subexponential time} in $(\bull,\E)$-free graphs, with no restrictions on the clique number.

\begin{restatable}{corollary}{corsubexp}
\label{cor:subexp}
    For every $k \geq 1$, \textsc{Max Partial $k$-Coloring} on $(\bull,\E)$-free instances with $n$ vertices can be solved in time $2^{\Oh(k \cdot \sqrt{n} \log^{3/2} n)}$.
\end{restatable}

We believe that the problem is actually polynomial-time solvable, and we state this as a conjecture.

\begin{restatable}{conjecture}{conj}
\label{conj}
    For every $k \geq 1$, \textsc{Max Partial $k$-Coloring} on $(\bull,\E)$-free graphs can be solved in polynomial time.
\end{restatable}

Let us remark that the situation here is similar to the case of \textsc{Odd Cycle Transversal} (and also  \textsc{Max Partial $k$-Coloring} and even slightly further) for $P_5$-free graphs. First, an algorithm with running time $n^{\Oh(k \omega)}$, i.e., similar to our \cref{thm:chair,thm:E}, was given by Chudnovsky, King, Pilipczuk, Rzążewski, and Spirkl~\cite{DBLP:journals/siamdm/ChudnovskyKPRS21}. Obtaining a polynomial-time algorithm was a well-known open problem in the area~\cite{DBLP:journals/siamdm/ChudnovskyKPRS21,DBLP:journals/algorithmica/DabrowskiFJPPR20,DBLP:journals/dagstuhl-reports/ChudnovskyPS19}, that was solved recently by Agrawal, Lima, Lokshtanov, Rzążewski, Saurabh, and Sharma~\cite{10.1145/3708544} for \textsc{Odd Cycle Transversal}.
Then the algorithm was generalized to \textsc{Max Partial $k$-Coloring} independently by Henderson, Smith-Roberge, Spirkl, and Whitman~\cite{henderson2024maximumkcolourableinducedsubgraphs} (even for a slightly wider class), and by Lokshtanov, Rzążewski, Saurabh, Sharma, and Zehavi~\cite{DBLP:journals/corr/abs-2410-21569} (even for a slightly more general problem).
However, the methods used there rely heavily on the structure of $P_5$-free graphs, and it is not clear to what extend they can be applied to \cref{conj}.

\section{Preliminaries}

For an integer $n$, by $[n]$ we denote the set $\{1,\ldots,n\}$.
All logarithms in the paper are of base 2.

We say that sets $A_1,\ldots,A_p$ form a \emph{partition} of $X$, if they are pairwise disjoint and their union is equal to~$X$. Note that we do not insist that sets $A_i$ are non-empty.

\paragraph{Graph theory.}
Let $G$ be a graph and let $X,Y$ be two disjoint subsets of $V(G)$.
We say that $X$ is \emph{complete} to $Y$, or that $X$ and $Y$ are \emph{complete to each other}, if $xy \in E(G)$ for any $x \in X$ and $y \in Y$.
Similarly, $X$ and $Y$ are \emph{anticomplete} if there are not edges between $X$ and $Y$.
If one of the sets $X,Y$ is a singleton, say $X = \{x\}$, we shortly say that $x$ is (anti)complete to $Y$.

Let $v$ be a vertex of $G$, and $X$ be a subset of $V(G)$.
By $N(v)$ we denote the set of all neighbors of $v$, and we define $N[v] = N(v) \cup \{v\}$. We also write $N_X(v)$ as a shorthand for $N(v) \cap X$.
We also define $N[X] = \bigcup_{v \in X} N[v]$ and $N(x) = N[X] \setminus X$.

As all subgraphs are in the paper are induced, we will sometimes identify them with their vertex sets.

\paragraph{Max Partial $k$-Coloring.}
We assume that all arithmetic operations on the values of the revenue function are performed in constant time.
Let $(G,\rev)$ be an instance of \textsc{Max Partial $k$-Coloring} and let $(X,c)$ be an optimum solution.
Note that without loss of generality we can assume that for every $v \in X$ it holds that $\rev(v,c(v))>0$.
Indeed, if there is $v \in X$ such that $\rev(v,c(v))=0$, we can obtain a solution of the same weight by removing $v$ from $X$ (and uncoloring it).

Thus in our algorithms we will use the convention that we never color a vertex using a color with zero revenue.
Consequently, in order to indicate that some $v$ \emph{cannot} receive color $i$, we will sometimes modify the revenue function by setting $\rev(v,i)=0$. We will call this operation \emph{forbidding color $i$ to $v$}.

Any instance $(G',\rev')$ of \textsc{Max Partial $k$-Coloring}, where $G'$ is an induced subgraph of $G$,
and $\rev'$ was obtained from $\rev$ by restricting the domain to $V(G')$ and (possibly) forbidding some colors to some vertices, will be called a \emph{subinstance} of $(G,\rev)$.

Slightly abusing the notation, we will keep denoting the modified revenue function (either by restricting to the vertex set of the current subinstance, or by forbidding some color) by $\rev$.

\paragraph{Gy\'arf\'as path argument.}
We will use the following known result, usually referred to as the \emph{Gy\'arf\'as path argument}, see~\cite{gyarfas2,DBLP:journals/algorithmica/BacsoLMPTL19,DBLP:journals/siamcomp/ChudnovskyPPT24}. 

\begin{theorem}\label{thm:gyarfas}
    Let $G$ be a connected graph on $n$ vertices, and let $v \in V(G)$.
    In polynomial time we can find an induced path $P$, starting at $v$, such that every component of $G - N[P]$ has at most $n/2$ vertices.
\end{theorem}

\section{Auxiliary results}
Before we proceed to proofs of main results, let us show some auxiliary results concerning the structure of considered graphs.

\subsection{Structural lemmas}

A \emph{fat path} (resp. \emph{fat cycle}) is a graph whose vertex set is partitioned into non-empty sets $V_1,V_2,\ldots,V_{r}$,
such that the sets $V_i,V_j$ are complete to each other if $j = i+1$ (resp. $j = i+1 \pmod r$) and anticomplete otherwise.
The number~$r$, i.e., the number of sets, is called the \emph{order} of the fat path (resp. cycle).

The following technical lemma is the key building block for the proof of \cref{thm:chair,thm:E}.
\begin{lemma}\label{lem:structureE}
    Let $G$ be a connected $(\bull,\E)$-free graph, containing a maximal induced path $P$ with at least 7 vertices.
    In polynomial time one can compute a partition of $V(G)$ into sets $R,D,T$ with the following properties.
    \begin{enumerate}
        \item $R$ is non-empty and induces a fat path of order at least 7 or a fat cycle of order at least 8,
        \item $D$ is complete to $R$, and separates $R$ and $T$,
        \item every component of $G - (R \cup D)$ is contained in one component of $G - N[Q]$.        
    \end{enumerate}
    Furthermore, if $G$ is $(\bull,\chair)$-free, then every component of $G - (R \cup D)$ is complete to some vertex in $D$.
\end{lemma}
\begin{proof} 
    If there is a vertex $v$ adjacent to both endvertices of $P$, and no other vertex of $P$, then we define $Q$ to be the induced cycle formed by vertices of $P$ and $v$.
    Otherwise, we define $Q=P$.
    We will show that $Q$ extends to the desired fat path (or fat cycle) $R$.

    Let us denote the consecutive vertices of $Q$ by $v_1,\ldots,v_r$; recall that $r \geq 7$.
    If $Q$ is a cycle, indices are computed modulo $r$. 
    We define the following subsets of $N(Q)$:
    \begin{align*}
        A_1= & \big \{w\in N(Q) ~|~  \{v_2\} \subseteq N_Q(w)\subseteq\{v_1,v_2\} \big \},\\
        A_r= & \big \{w\in N(Q) ~|~ \{v_{r-1}\} \subseteq N_Q(w)\subseteq\{v_{r-1},v_{r}\} \big \},\\        
        D= & \big \{w\in N(Q) ~|~ N_Q(w)=Q \big \}.
    \end{align*}
    Moreover, for any $i \in \{2, \ldots, r-2\}$ (or for $i \in \{1, \ldots, r\}$, if $Q$ is a cycle) we define:
    \begin{align*}
        B_i= & \big \{w\in N(Q) ~|~ N_Q(w)=\{v_{i-1},v_{i+1}\}  \big  \},\\        
        C_i= & \big \{w\in N(Q) ~|~ N_Q(w)=\{v_{i-1},v_{i},v_{i+1}\} \big \}.
    \end{align*}
Additionally, let us define $B = \bigcup_{i} B_i$ and $C = \bigcup_{i} C_i$. Obviously $B_i,B_j$ (resp. $C_i,C_j$) are disjoint whenever $i \neq j$.
\begin{claim}\label{clm:NQ}
    Sets $A_1,A_r,B,C,D$ form a partition of $N(Q)$.
\end{claim}
\begin{claimproof}
    Clearly the sets are pairwise disjoint. We need to show that any $w \in N(Q)$ belongs to one of sets $A_1,A_r,B,C,D$.
    
    Let $\ell$ be the largest integer such that $w$ has $\ell$ consecutive neighbors on $Q$ (again if $Q$ is a cycle, then we compute indices modulo $r$; in particular, we treat $v_1$ and $v_r$ as consecutive).
    If $\ell = r$, i.e., $N_Q(w) = Q$, then $w \in D$.
    So assume otherwise and let a longest sequence of consecutive neighbors of $w$ on $Q$ be $v_i, \ldots, v_{i+\ell-1}$. We will consider several cases.    
        
    Suppose first that $\ell \geq 4$. Then by symmetry we assume $v_{i-1}w \notin E(G)$ and the set of the vertices $\{v_{i-1}, v_i, v_{i+1}, w, v_{i+3}\}$ induces the bull (see \cref{pic:claim_partition_a}).  

    Suppose now $\ell = 2$. Then $w$ has two consecutive neighbors $v_i, v_{i+1}$ on $Q$.
    If $Q$ is a cycle or $Q$ is a path but neither $v_i$ nor $v_{i+1}$ are its endvertices, then the set of vertices $\{v_{i-1}, v_i, w, v_{i+1}, v_{i+2}\}$ induces the bull (see \cref{pic:claim_partition_b}).
    Thus, suppose that $Q$ is a path with $v_i$ being its endvertex (the case if $v_{i+1}$ is an endvertex is symmetric).
    Then either $N_Q(w)=\{v_1, v_2\}$ and thus $w \in A_1$, or $w$ has some neighbors among $\{v_4,\ldots,v_r\}$.
    If there is a neighbor $v_k$ of $w$, where $k > 4$, then the set of vertices $\{v_3, v_2, v_1, w, v_k\}$ induces the bull (see \cref{pic:claim_partition_c}).
    Otherwise, the only neighbors of $w$ in $\{v_4,\ldots,v_r\}$ is $v_4$ and the set of vertices $\{v_1, w, v_4, v_5, v_6, v_3\}$ induces $\E$ (see \cref{pic:claim_partition_d}).

    \begin{figure}[t]
    \centering    
    \begin{subfigure}{0.22\textwidth}
    \centering  
    \includegraphics[width=1.1\linewidth]{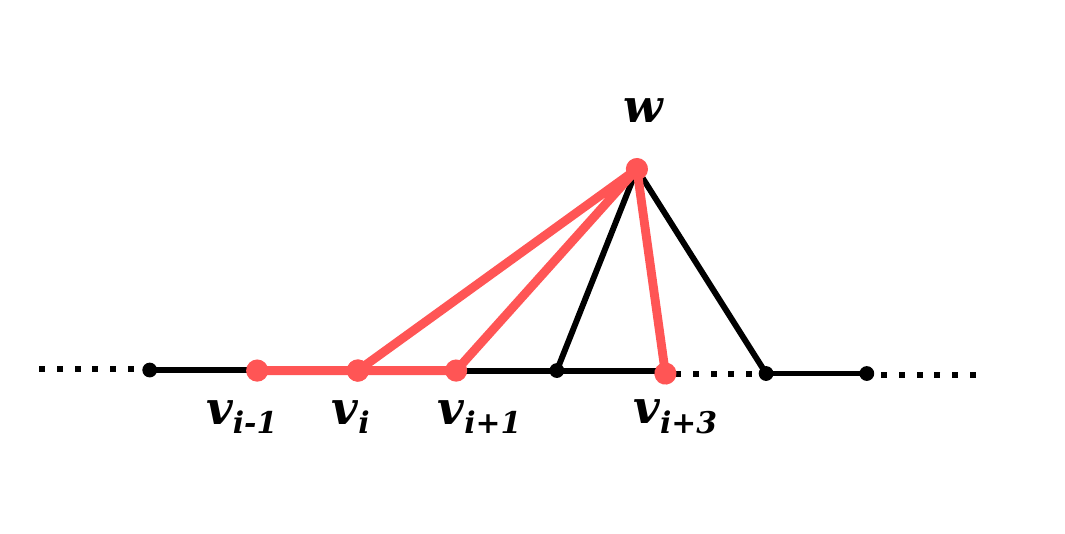} 
    \caption{\label{pic:claim_partition_a}}
    \end{subfigure}
    \begin{subfigure}{0.22\textwidth}
    \centering  
    \includegraphics[width=1.1\linewidth]{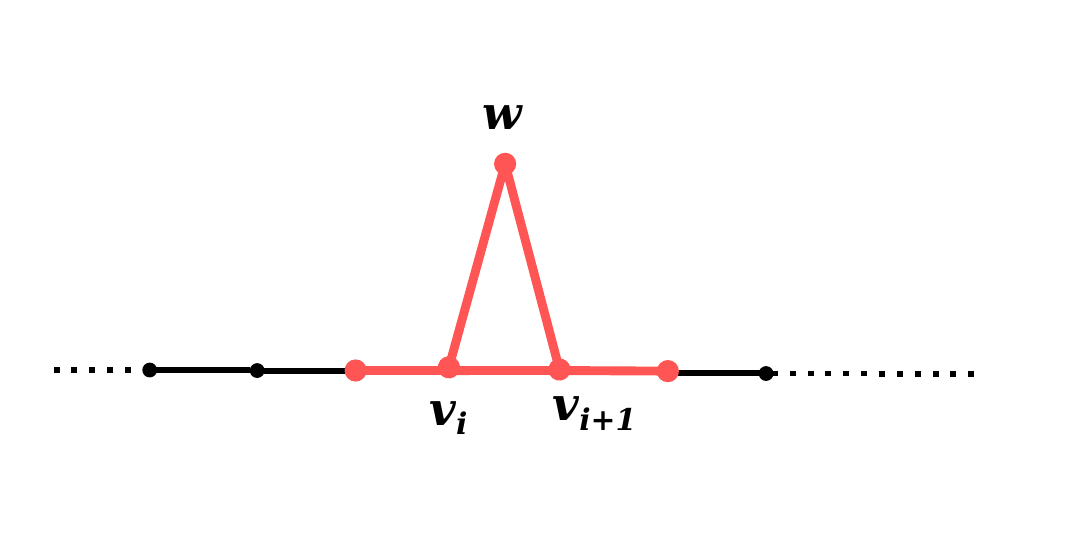}
    \caption{\label{pic:claim_partition_b}}
    \end{subfigure}
    \begin{subfigure}{0.22\textwidth}
    \centering  
    \includegraphics[width=1.1\linewidth]{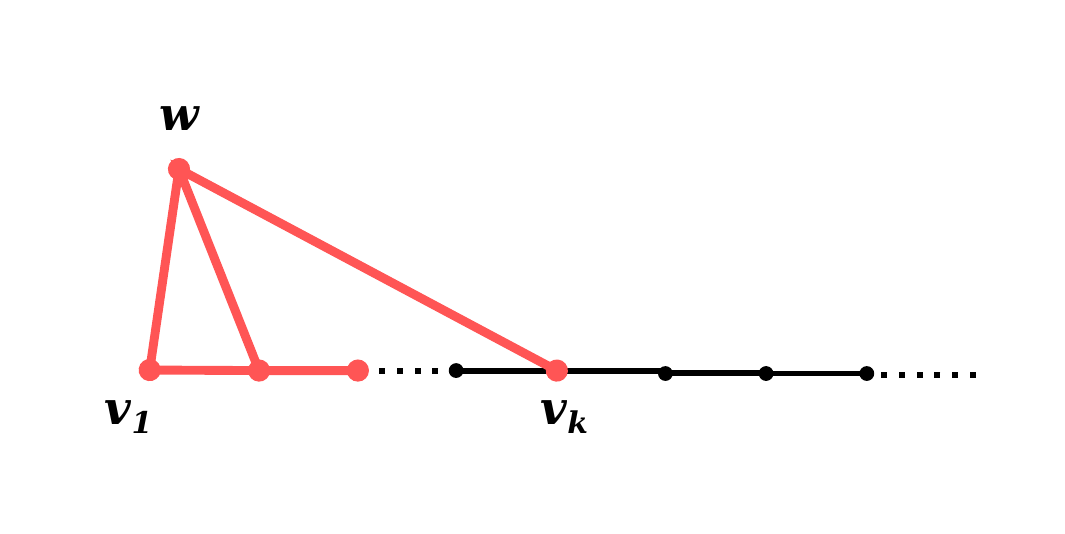}
    \caption{\label{pic:claim_partition_c}}    
    \end{subfigure}
    \begin{subfigure}{0.22\textwidth}
    \centering  
    \includegraphics[width=1.1\linewidth]{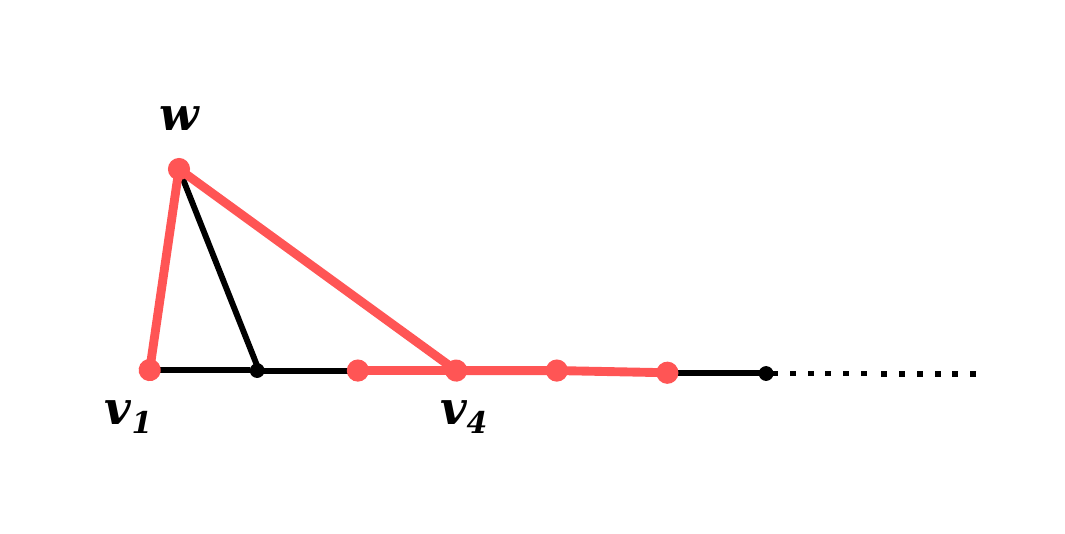}
    \caption{\label{pic:claim_partition_d}}    
    \end{subfigure}
    \caption{Induced subgraphs constructed in the proof of \cref{clm:NQ}, the case that $\ell>1$.}\label{pic:claim_partition1}
\end{figure}

    Next, suppose that $\ell =1$, i.e., $w$ has no two consecutive neighbors on $Q$.
    Let us consider possible cases.
    Suppose first that $w$ has (at least) three neighbors $v_{a}, v_{b}, v_{c}$ on $Q$, where $a < b -1< c -2$, such that one of the following cases applies:
    \begin{enumerate}[(i)]
        \item $Q$ is a cycle, or
        \item $Q$ is a path and neither $v_{a}$ nor $v_{c}$ is an endvertex of $Q$, or
        \item $Q$ is a path, one of $v_a,v_c$ is its endvertex, and either $b \geq a+2$ or $c \geq c+2$.
    \end{enumerate} 
    Then in the set $\{v_{a-1}, v_{a+1} v_{b-1}, v_{b+1}, v_{c-1}, v_{c+1}\}$ we can always find two vertices extending the induced star $\{v_a, v_b, v_c, w\}$ to $\E$ (here we slightly abuse the notation, as some of the vertices of that set might not exist if $Q$ is a path, but the two that form an induced copy of $\E$ always exist; see \cref{pic:claim_partition2_a}).
    Note that one of cases (i), (ii), (iii) applies whenever $w$ has at least four neighbors on $Q$.
    If $w$ has exactly three neighbors $v_a,v_b,v_c$ on~$Q$, the only uncovered case is when $Q$ is a path and, by symmetry, $a=1$, $b=3$, and $c=5$.
    Then the set of vertices $\{v_{1}, w, v_5, v_4, v_6,v_7\}$ induces $\E$ (see \cref{pic:claim_partition2_b}). 
    
    Now suppose that $w$ has exactly two non-consecutive neighbors $v_i, v_j$ on $Q$, where $i < j-1$.
    If $j=i+2$ then $w\in B_{i+1}$.    
    Otherwise, one of sets $\{v_{i-2}, v_{i-1}, v_i, v_{i+1}, v_{i+2}, w \}$, $\{v_{j-2}, v_{j-1}, v_j, v_{j+1}, v_{j+2}, w \}$ (see \cref{pic:claim_partition2_c}), $\{v_{i+2}, v_{i+1}, v_i, w, v_j, v_{i-1}\}$, $\{v_{j-2}, v_{j-1}, v_j, w, v_i,v_{j+1}\}$ (see \cref{pic:claim_partition2_d}) induces~$\E$ . (Let us recall that both $v_i$, $v_j$ cannot be endvertices of the path.)  

    \begin{figure}[t]
    \centering    
    \begin{subfigure}{0.22\textwidth}
    \centering  
    \includegraphics[width=1.1\linewidth]{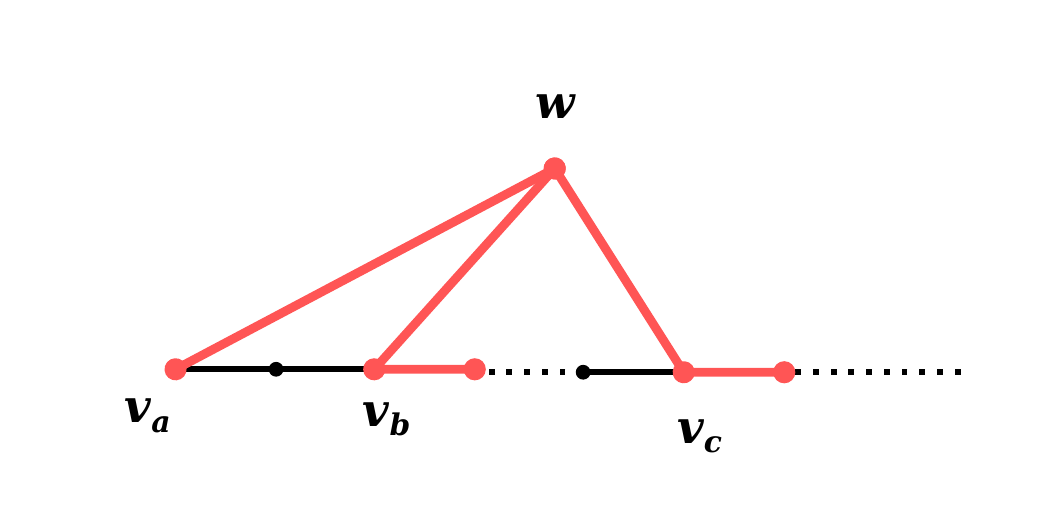} 
    \caption{\label{pic:claim_partition2_a}}
    \end{subfigure}
    \begin{subfigure}{0.22\textwidth}
    \centering  
    \includegraphics[width=1.1\linewidth]{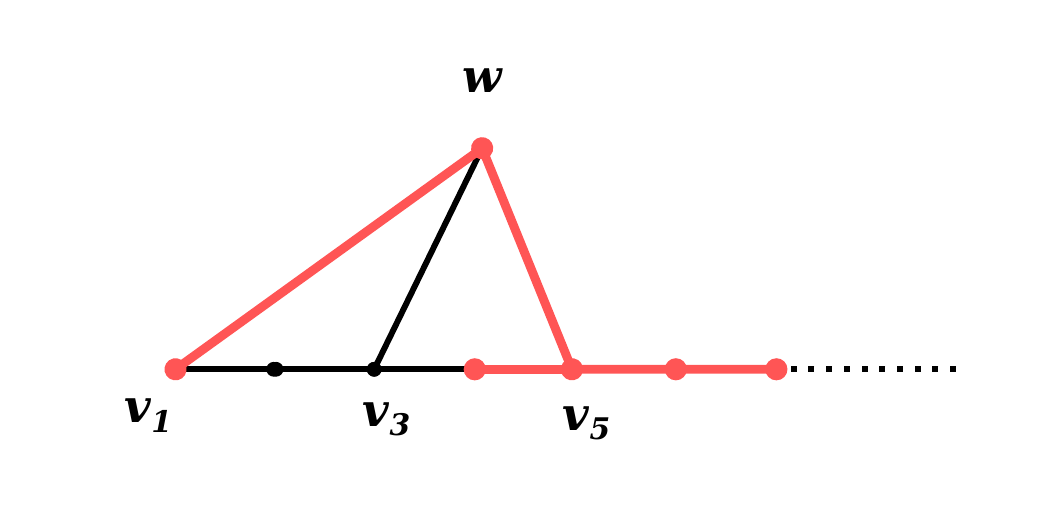}
    \caption{\label{pic:claim_partition2_b}}
    \end{subfigure}
    \begin{subfigure}{0.22\textwidth}
    \centering  
    \includegraphics[width=1.1\linewidth]{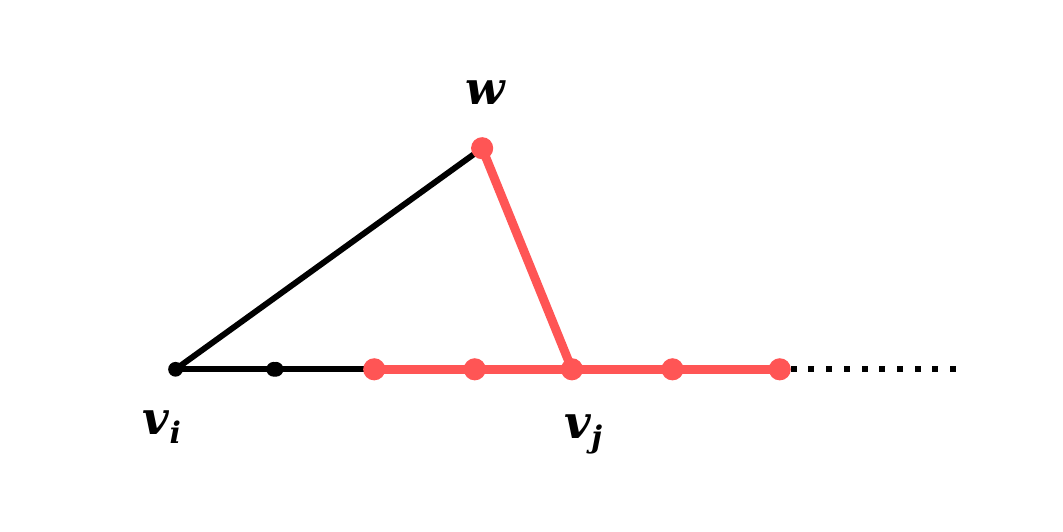}
    \caption{\label{pic:claim_partition2_c}}    
    \end{subfigure}
    \begin{subfigure}{0.22\textwidth}
    \centering  
    \includegraphics[width=1.1\linewidth]{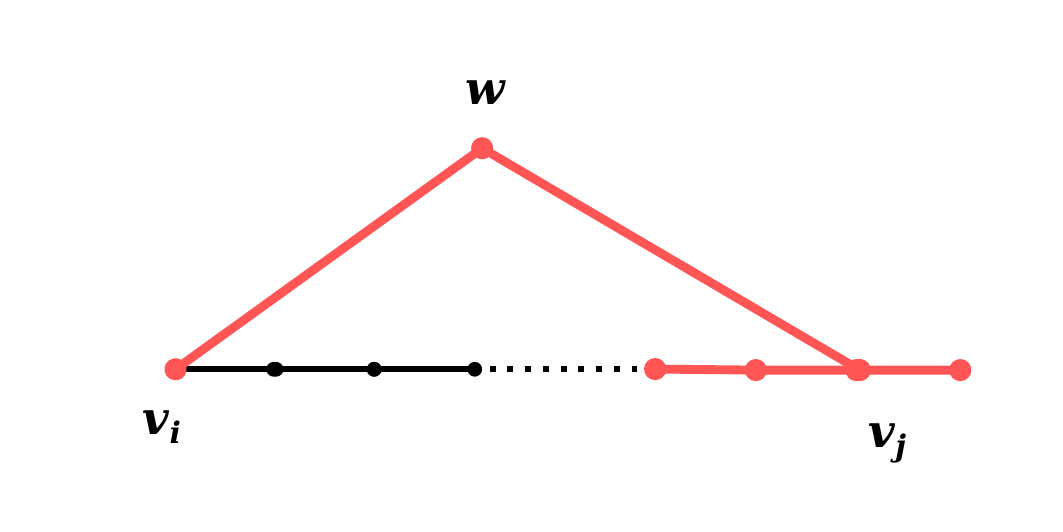}
    \caption{\label{pic:claim_partition2_d}}    
    \end{subfigure}
    \caption{Induced subgraphs constructed in the proof of \cref{clm:NQ}, the case that $\ell=1$.}\label{pic:claim_partition2}
\end{figure}
    
    Finally, if $w$ has only one neighbor $v_i$ on $Q$, then either $Q$ is a path and $i \in \{2, r-1\}$ and thus $w \in A_1 \cup A_r$, or the set of vertices $\{v_{i-2}, v_{i-1}, v_i, v_{i+1}, v_{i+2}, w\}$ induces $\E$.   

    We showed that either $\ell=r$ (so $w \in D$) or $\ell=2$ and $w \in A_1 \cup A_r$, or $\ell=1$ and $w \in A_1 \cup A_r \cup B$, or $\ell=3$.
    What still needs to be proven is that if $\ell=3$, then $w$ has exactly three neighbors on $Q$ and thus $w\in C$.  Suppose now that $\{v_i, v_{i+1}, v_{i+2}\} \subseteq N_Q$ and there is another $v_j$ such that $\left.wv_j \in E(G)\right.$ and $j \notin \{i-1, i+3\}$.
    By symmetry, assume that $j>i+3$.
    If $v_i$ is not an endvertex of the path $Q$ or $j>i+4$, then one of the sets $\{v_{i-1}, v_i, v_{i+1}, w, v_j\}$, $\{v_{i+3}, v_{i+2}, v_{i+1}, w, v_j\}$ induces the bull.
    Otherwise, if $Q$ is a path, $i=1$ and $j=5$, then the set of vertices $\{v_1, w, v_5, v_6, v_7, v_4\}$ induces $\E$.
\end{claimproof}

 For each $i\in\{2, \ldots, r\}$ if $Q$ is a path, or for each $i \in [r]$ if $Q$ is a cycle, we define $V_i=\{v_i\} \cup B_{i}\cup C_{i}$.
 If $Q$ is a path, we additionally define $V_1=\{v_1\}\cup A_1$ and $V_r=\{v_r\}\cup A_r$.

\begin{claim}\label{clm:fatpath}
    $G[V_1 \cup \ldots \cup V_r]$ is a fat path or a fat cycle with consecutive sets $V_1,\ldots,V_{r}$.
\end{claim}
    \begin{claimproof}
    We need to show that for $i \neq j$ the sets $V_i, V_j$ are complete to each other if $j=i+1$ (or $j=i+1 \pmod r$ if $Q$ is a cycle) and anticomplete otherwise.

    Assume $i < j$ and pick any $w_i \in V_i$ and $w_j \in V_j$.
    If $w_i = v_i$ or $w_j = w_j$ then the claim is clear by the definition of sets  $V_i,V_j$.
    So suppose otherwise.

    Suppose first that $j = i+1$ and $w_iw_j \notin E(G)$.
    Either $i \geq 4$ or $j \leq r-3$; by symmetry assume the former.
    If $w_j \in B_i$, then $\{v_{j+3},v_{j+2},v_{j+1},v_j,w_i,w_j\}$ induces $\E$.
    If $w_j \in C_i$, then $\{v_{j+2},v_{j+1},w_j, v_j, w_i\}$ induces the bull (see \cref{pic:claim_fatcycle_a} and \cref{pic:claim_fatcycle_b}).

    Now assume that $j > i+1$ and $w_iw_j \in E(G)$. Suppose first that $w_j \in B_j$. 

    If $j > i+3$, then $\{w_i, w_j, v_{j-1}, v_{j-2}, v_{j-3}, v_j\}$ induces $\E$. If $j = i+3$ or $j=i+2$, then either $i \geq 4$ or $j \leq r-3$; assume the former by symmetry. Then $\{w_i,w_j,v_{j+1},v_{j+2},v_{j+3}, v_j\}$ induces $\E$ (see \cref{pic:claim_fatcycle_c}).

    Suppose now $w_j \in C_j$. If $j>i+3$, then $\{w_i, w_j, v_j, v_{j-1}, v_{j-2}\}$ induces the bull (see \cref{pic:claim_fatcycle_d}). Otherwise assume by symmetry $j \leq r-3$. Then $\{w_i, w_j, v_j, v_{j+1}, v_{j+2}\}$ induces the bull. 
        This completes the proof of the claim.
    \end{claimproof}

    \begin{figure}[t]
    \centering    
    \begin{subfigure}{0.22\textwidth}
    \centering  
    \includegraphics[width=1.1\linewidth]{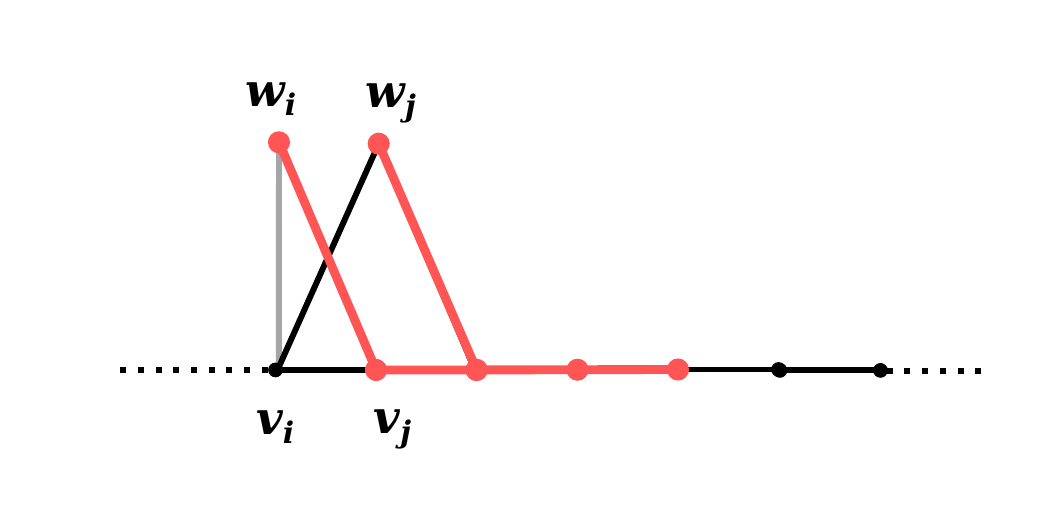} 
    \caption{\label{pic:claim_fatcycle_a}}
    \end{subfigure}
    \begin{subfigure}{0.22\textwidth}
    \centering  
    \includegraphics[width=1.1\linewidth]{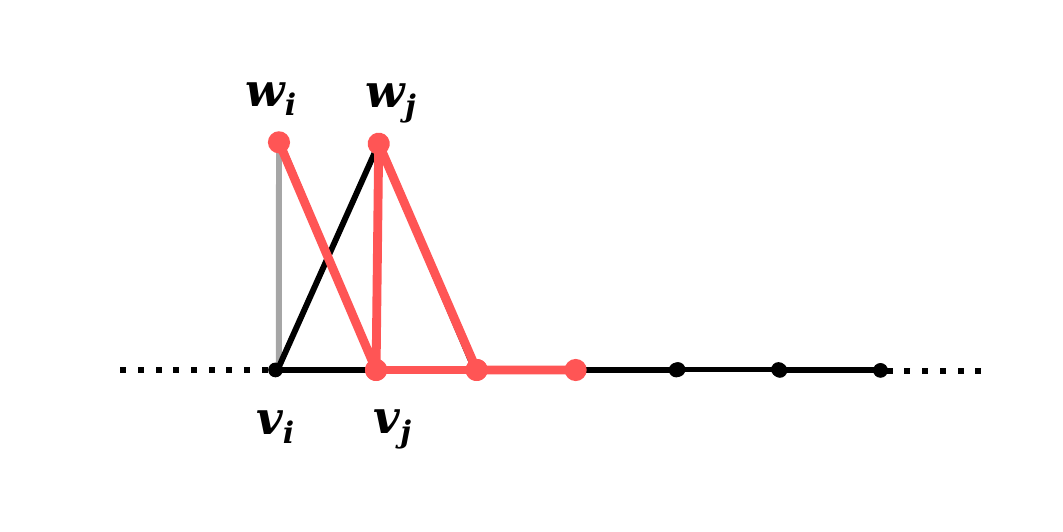}
    \caption{\label{pic:claim_fatcycle_b}}
    \end{subfigure}
    \begin{subfigure}{0.22\textwidth}
    \centering  
    \includegraphics[width=1.1\linewidth]{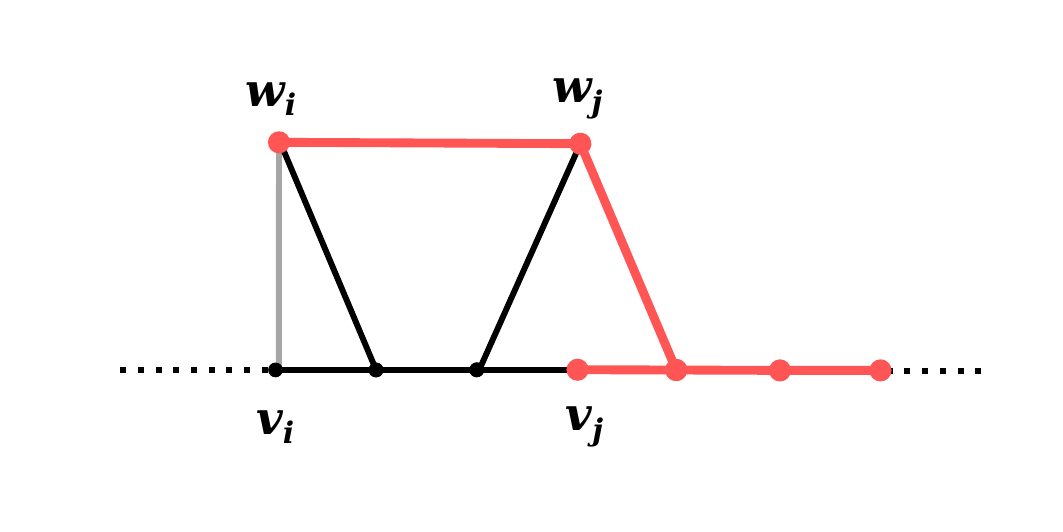}
    \caption{\label{pic:claim_fatcycle_c}}    
    \end{subfigure}
    \begin{subfigure}{0.22\textwidth}
    \centering  
    \includegraphics[width=1.1\linewidth]{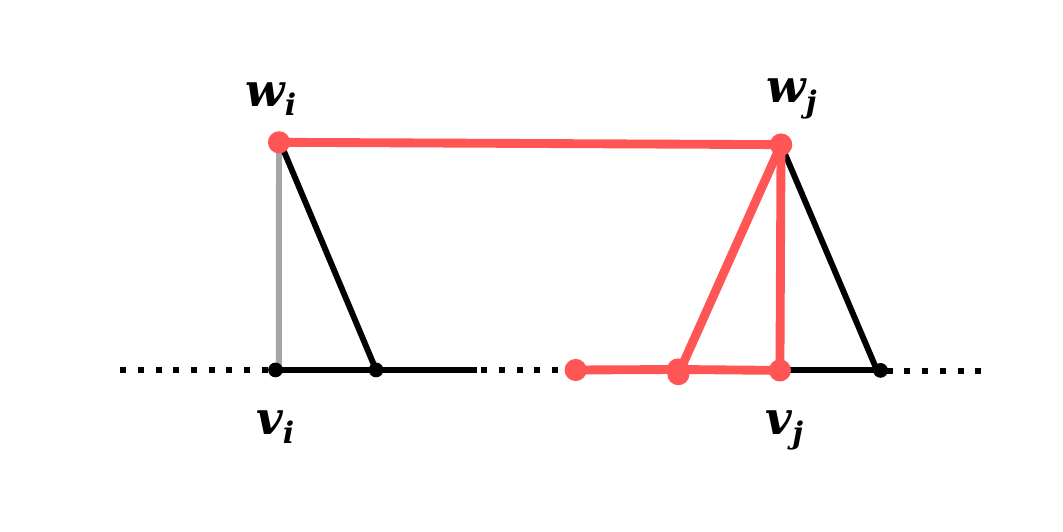}
    \caption{\label{pic:claim_fatcycle_d}}    
    \end{subfigure}
    \caption{Induced subgraphs constructed in the proof of \cref{clm:fatpath}.}\label{pic:claim_fatcycle}
\end{figure}

    So let $R=G[V_1 \cup \ldots \cup V_r]$.
    Now let us argue about the second statement of the lemma.

    \begin{claim}\label{clm:D_complete_to_R}
        $D$ is complete to $R$.        
    \end{claim}
    \begin{claimproof}
    By the definition, $D$ is complete to $Q$. Thus consider $w_i\in V_i\setminus\{v_i\}$ for some $i$ and $d \in D$. 
    Pick any $j \in [r] \setminus \{i-2,i-1,i,i+1,i+2\}$; it exists since $r \geq 8$.
    Note that if $w_id \notin E(G)$, then one of the sets $\{w_i,v_{i+1},v_{i+2},d,v_j\}$ or $\{w_i,v_{i-1},v_{i-2},d,v_j\}$ is well-defined and  induces the bull, a contradiction (see \cref{pic:claim_a}).
    \end{claimproof}

    \begin{claim}\label{clm:only_D_has_N2}
        There are no edges between $R$ and $V(G) \setminus (R \cup D)$.
    \end{claim}
    \begin{claimproof}
        For contradiction, suppose that there is $u \notin R \cup D$ and $w \in R$ that are adjacent.
        Note that $w \in V_i \setminus \{v_i\}$, for some $i \in \{1, \ldots, r\}$, as $u$ has no edges to $Q$ by \cref{clm:NQ}.
        Again we have several cases, but either one of sets $\{u,w,v_i,v_{i+1},v_{i+2}\}$, $\{u,w,v_i,v_{i-1},v_{i-2}\}$ induces the bull (if $wv_i \in E(G)$;  see \cref{pic:claim_b}), or
        one of sets $\{u,w,v_{i+1},v_i,v_{i+2},v_{i+3}\}$, $\{u,w,v_{i-1},v_i,v_{i-2},v_{i-3}\}$ (if $wv_i \notin E(G)$; see \cref{pic:claim_c}).        
    \end{claimproof}
        
    The third statement of the lemma is immediate, as $N[Q] \subset R \cup D$, and thus every component of $G- (R \cup D)$ is contained in one component of $G - N[Q]$. 
    Now let us show the stronger property obtained in the case that $G$ is $(\bull, \chair)$-free.
    
\begin{claim}\label{clm:bc_comps_compl_to_D}
     If $G$ is $(\bull, \chair)$-free, then every component of  $G - (R \cup D)$ is complete to some vertex of $D$.
\end{claim}
\begin{claimproof}
    Suppose otherwise and let $C$ be a component $C$ of $G - (R \cup D)$ which is not complete to any vertex of $D$. 
    Since $G$ is connected and $D$ separates $R$ and $G-(R \cup D)$, there exist two adjacent vertices $u, u' \in C$ such that $u$ is adjacent to some vertex $d \in D$ and $du' \notin E(G)$. But then the set of vertices $\{u', u, d, v_1, v_3\}$ induces the chair (see \cref{pic:claim_d}).
\end{claimproof}

    \begin{figure}[t]
    \centering    
    \begin{subfigure}{0.22\textwidth}
    \centering  
    \includegraphics[width=1.1\linewidth]{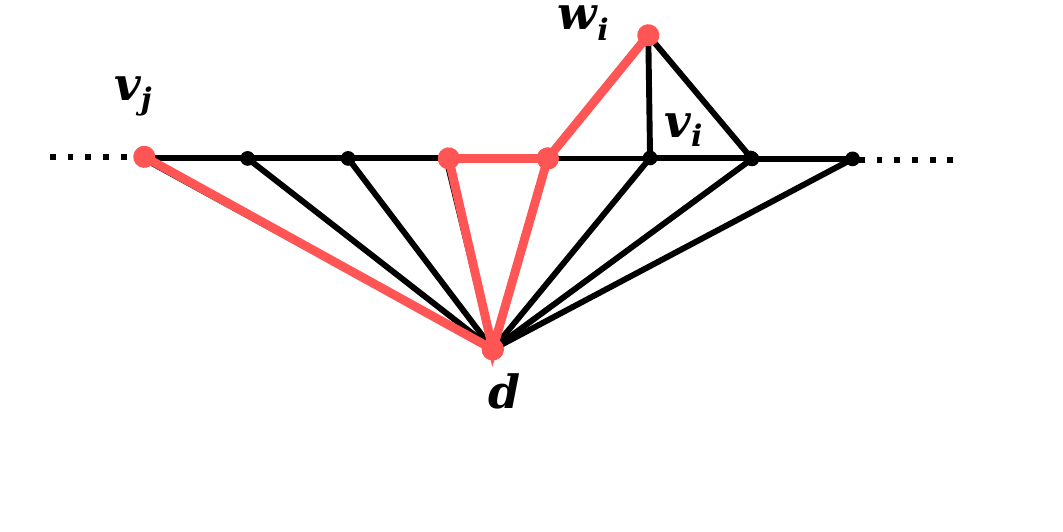} 
    \caption{\label{pic:claim_a}}
    \end{subfigure}
    \begin{subfigure}{0.22\textwidth}
    \centering  
    \includegraphics[width=1.1\linewidth]{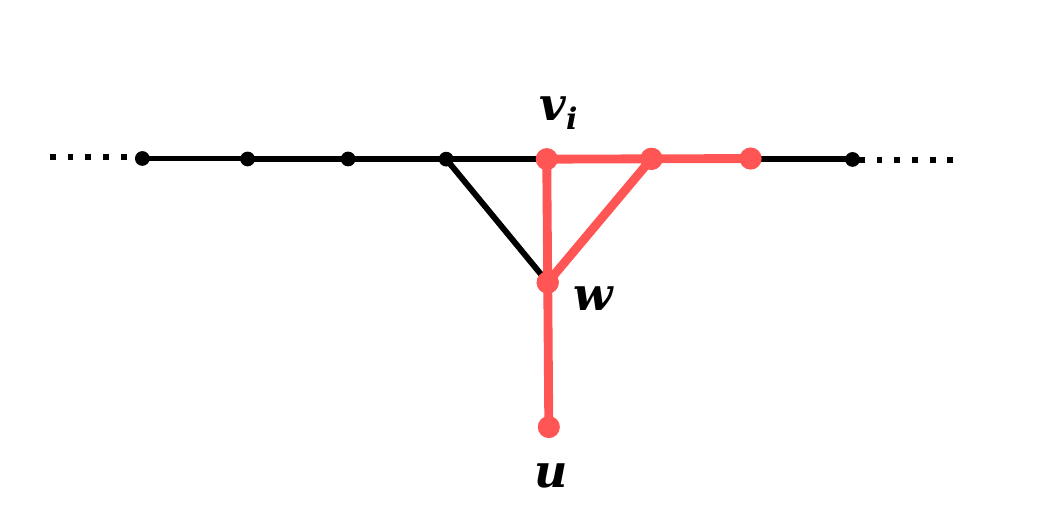}
    \caption{\label{pic:claim_b}}
    \end{subfigure}
    \begin{subfigure}{0.22\textwidth}
    \centering  
    \includegraphics[width=1.1\linewidth]{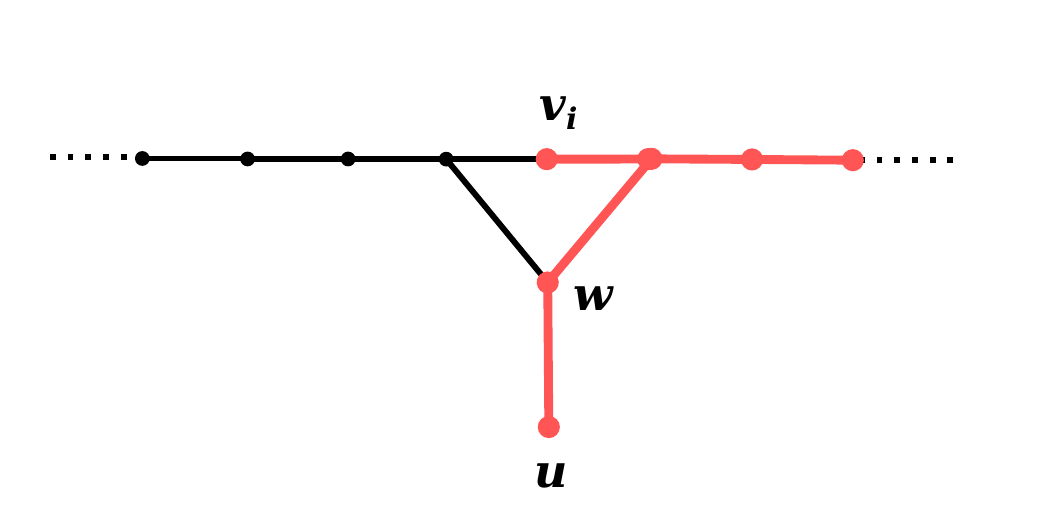}
    \caption{\label{pic:claim_c}}    
    \end{subfigure}
    \begin{subfigure}{0.22\textwidth}
    \centering  
    \includegraphics[width=1.1\linewidth]{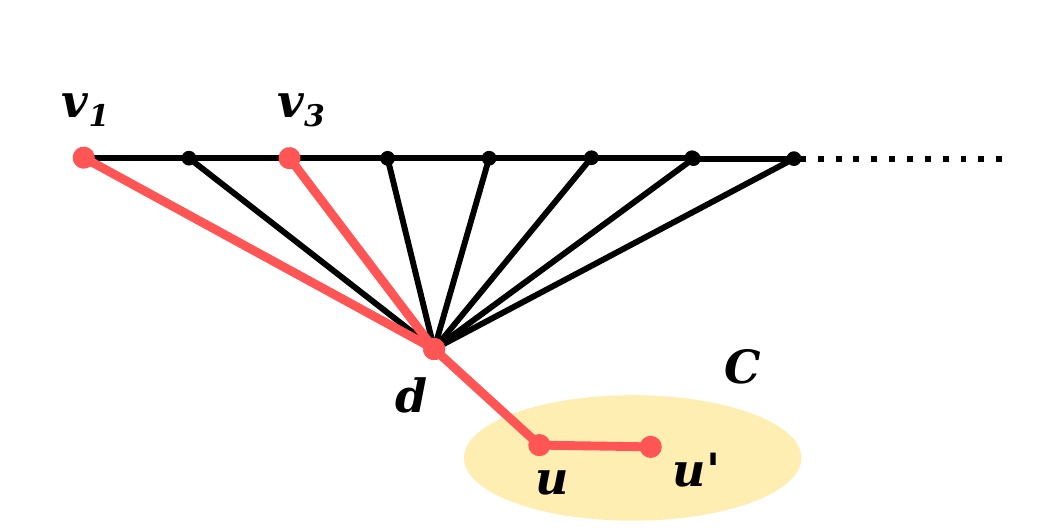}
    \caption{\label{pic:claim_d}}    
    \end{subfigure}
    \caption{Induced subgraphs constructed in proofs of  \cref{clm:D_complete_to_R,clm:only_D_has_N2,clm:bc_comps_compl_to_D}.}\label{pic:claim}
\end{figure}
    
    To complete the proof, we observe that the partition of $G$ into $R,D,T$ can easily be found in polynomial time.
\end{proof}

\begin{lemma}\label{lem:structureCh}
    Let $G$ be a $(\bull,\chair)$-free graph on at least three vertices, and let $P$ be a maximal induced path on at least 3 vertices.
    For every component $C$ of $G- N[P]$ there exists a vertex $v \in N(P)$ that is complete to $C$.    
\end{lemma}
\begin{proof}
    Let us suppose that there exists a component $C$ of $G - N[P]$ such that $C$ is not complete to any vertex of $N[P]$. Let $uu'$ be an edge in $C$ such that $wu \in E(G)$ and $wu' \notin E(G)$ for some $w \in N[P]$. Let us consider possible cases. 

    If $w$ is adjacent to every vertex of $P$, then the set of vertices $\{u', u, w, v_1, v_3\}$ induces the chair. If $w$ is adjacent to two consecutive vertices on $P$, but not to the whole $P$, we can always find three consecutive vertices $v_i, v_{i+1}, v_{i+2}$ such that $wv_{i+1}\in E(G)$ and exactly one of $v_i, v_{i+2}$ is a neighbor of $w$. Then the set of vertices $\{v_i, v_{i+1}, v_{i+2}, w, u\}$
    induces the bull. Finally, if $w$ has no consecutive neighbors on $P$, then for its any neighbor $v_k$ the set of vertices $\{v_{k-1}, v_k, v_{k+1}, w, u\}$ induces the chair. 
\end{proof}

\begin{lemma}\label{lem:minimal}
    Let $G$ be an $(\bull,\E)$-free graph, $X$ be a subset of $V(G)$, and $c$ be a proper $k$-coloring of $G[X]$.
    Let $S,T$ be disjoint subsets of $V(G)$, and let $v \in V(G) \setminus (S \cup T)$,
    which is complete to $S$ and anticomplete to $T$.
    For $i \in [k]$, let $S^i$ be an inclusion-wise minimal subset of $S \cap X \cap c^{-1}(i)$,
    such that every vertex from $T$ that has a neighbor $S \cap X \cap c^{-1}(i)$, has a neighbor in $S'$.
    Then $|S^i| \leq 2$.
\end{lemma}
\begin{proof}
    Let us suppose that for some color $i$ the minimal subset $S^i$ satisfying the claim contains at least three vertices $u_1, u_2, u_3$. Since every vertex of the set $S^i$ is colored with $i$ by proper coloring, no two of $u_1, u_2, u_3$ are adjacent. From the minimality of the set $S^i$, every vertex $u_i$ has a private neighbor $w_i$ in $T$. If there exist every edge between $w_1, w_2, w_3$, then the set of vertices $\{u_1, w_1, w_2, w_3, u_3\}$ induces the bull (see \cref{pic:minsub_a}). Otherwise assume without losing generality $w_1, w_3$ are not adjacent. Then the set of vertices $\{w_1, u_1, v, u_3, w_3, u_2\}$ induces $\E$ (see \cref{pic:minsub_b}).

\begin{figure}[b]
    \centering    
    \begin{subfigure}{0.45\textwidth}
    \centering  
    \includegraphics[width=0.8\linewidth]{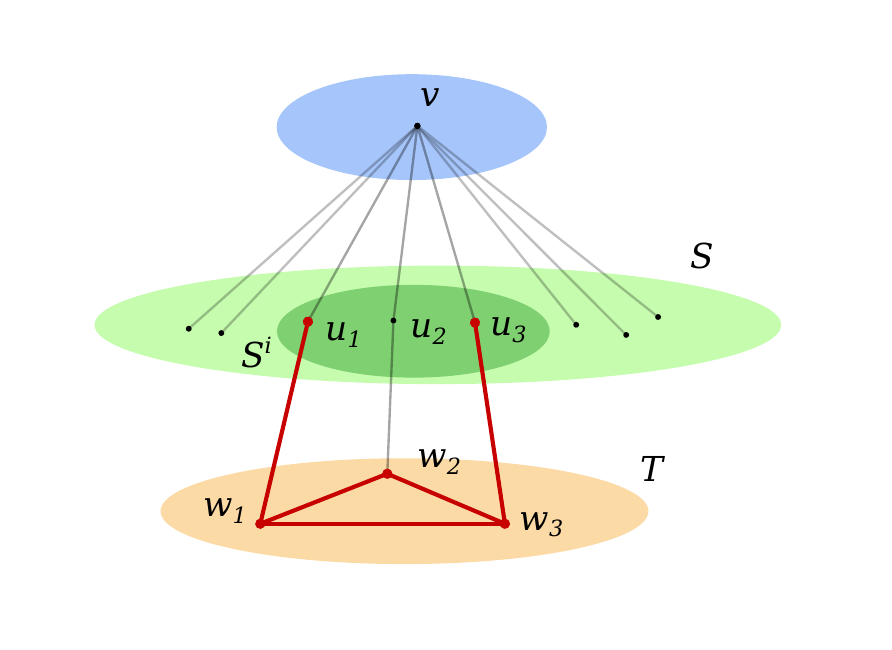}
    \caption{\label{pic:minsub_a}}
    \end{subfigure}
    \begin{subfigure}{0.45\textwidth}
    \centering  
    \includegraphics[width=0.8\linewidth]{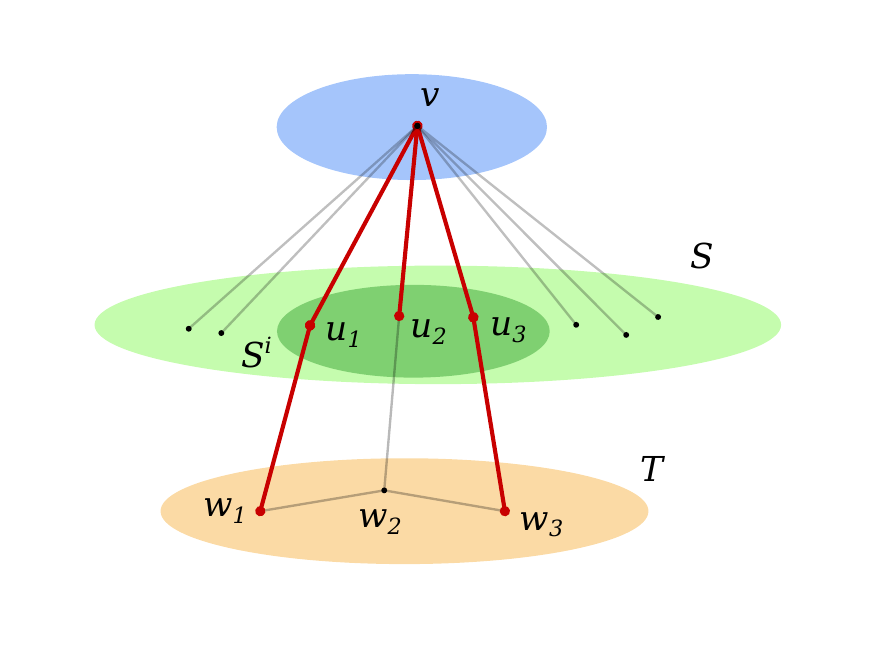} 
    \caption{\label{pic:minsub_b}}
    \end{subfigure}
    \caption{Induced subgraphs constructed in the proof of \cref{lem:minimal}.}\label{pic:minsub}
\end{figure}
\end{proof}

\subsection{Dealing with fat paths and fat cycles}

\begin{lemma}\label{lem:fatpath} 
    Let $k \geq 1$ and let $(G,\rev)$ be an instance of  \textsc{Max Partial $k$-coloring}, where $G$ is a fat path or a fat cycle with consecutive sets $V_1,\ldots,V_r$.
    Suppose that for each $i \in [r]$, any subinstance on vertex set $V_i$ can be solved in time $\lambda$.
    Then $(G,\rev)$ can be solved in time $\Oh(6^k \cdot \lambda \cdot r)$.
\end{lemma}
\begin{proof}
    The algorithm is a straightforward dynamic programming. We describe how to compute the value of an optimum solution; the solution itself can be extracted in standard way.
    
    First, for each $i \in [r]$ and each $B \subseteq [k]$,
    we want to compute $\textsf{OPT}[i,B]$, which is the value of an optimum solution on $G[V_i]$ that does not use colors from $[k] \setminus B$.
    In order to compute $\textsf{OPT}[i,B]$, we call the assumed algorithm for the subinstance $(G[V_i],\rev_B)$, where $\rev_B$ is obtained from $\rev$ by restricting its domain to $V_i$ and forbidding every color \emph{not} in $B$ to every vertex.
    Thus, all values of $\textsf{OPT}[\cdot,\cdot]$  can be computed in total time $\Oh(2^k \cdot r \cdot \lambda)$.

    Now, for each $i \in [r]$ and each $B \subseteq [k]$,
    we want to compute $\textsf{PRE}[i,B]$, which is the value of an optimum solution on $G[V_1 \cup \ldots \cup V_i]$ such that colors from $[k] \setminus B$ do not appear on $V_i$.
    Clearly the optimum solution is the maximum entry $\textsf{PRE}[r,B]$, over all sets $B$.

    Let us first consider the case that $G$ is a fat path.
    Clearly $\textsf{PRE}[1,B] = \textsf{OPT}[1,B]$ for every $B$.
    Now, having computed all values $\textsf{PRE}[i-1,\cdot]$,
    we set
    \[
        \textsf{PRE}[i,B] = \max_{B' \subseteq [k] \setminus B} \textsf{PRE}[i-1,B'] + \textsf{OPT}[i,B].
    \]
    Thus, computing all values of $\textsf{PRE}[\cdot,\cdot]$ takes total time $\Oh(3^k \cdot r)$.

    The case that $G$ is a fat cycle is very similar. We exhaustively guess the set $B_1$ of colors that are allowed to appear on $V_1$, which gives $2^k$ branches.
    We set $\textsf{OPT}[r,B]$ to 0 whenever $B$ intersects $B_1$, as such colors cannot appear on $V_r$.
    Then we proceed as in the case of a fat path.
    The complexity in this case is $\Oh(2^k \cdot 3^k \cdot r)=\Oh(6^k \cdot r)$.
    The total complexity is thus upper-bounded by $\Oh(6^k \cdot r \cdot \lambda)$.  
\end{proof}

Let us remark that another way of dealing with fat paths/cycles and similar structures would be to employ \emph{modular decomposition}~\cite{DBLP:conf/icalp/TedderCHP08}. However, this would require introducing some additional level of technical complication, in particular, considering a more general coloring problem.
Thus we decided to solve the problem by hand.

\section{Proof of \cref{thm:chair,thm:E}}
For both results we use \emph{exactly the same} algorithm, the only difference is in the complexity analysis.

Let $(G,\rev)$ be the input instance of \textsc{Max Partial $k$-coloring}, where $G$ is $(\bull,\E)$-free, has $n$ vertices and clique number at most $\omega$.
If $G$ is disconnected, we consider each component separately, so assume otherwise.
The proof proceeds by the induction on $n + \omega$. Note that $n \geq \omega$.

One base case is $\omega =1$, i.e., $G$ is a one-vertex graph -- this case is trivial.
Furthermore, we will assume that $n$ is larger than some constant $n_0$ that follows implicitly from the complexity analysis given later in the proof.
Of course if $n$ is at most $n_0$, we can solve the problem using brute force.
Thus suppose that $\omega \geq 2$ and $n$ is sufficiently large. In particular, $n > \omega$, so $G$ is not a clique.
We inductively assume that we have an algorithm that solves every instance with $n' \leq n$ vertices and clique number  $\omega' \leq \omega$, where $n' + \omega' < n + \omega$.

Fix some (unknown) optimum solution $(X,c)$. In the algorithm we will perform of series of exhaustive guesses and branch into corresponding options. We will ensure that always (at least) one branch is ``correct,'' i.e., it corresponds to the structure of $(X,c)$. This way we are guaranteed that our algorithm indeed finds an optimum solution (though not necessarily $(X,c)$).

\paragraph{Building the Gy\'arf\'as path.}
Since $G$ is not a clique, there is a vertex $v$ such that $N[v] \neq V(G)$.
We start with calling \cref{thm:gyarfas} for $G$ and $v$; let $P$ be the returned path.
We greedily extend $P$ (keeping it an induced path): if there is a vertex $v$ in $G - P$ whose neighborhood in $P$ is exactly one vertex, which is an endvertex of~$P$, then we append $v$ to $P$.
We perform this step exhaustively; let us keep denoting the final path by $P$.
Note that since $N[v] \neq V(G)$, $P$ has at least three vertices.
Furthermore, we observe that is still has the property from \cref{thm:gyarfas}, i.e., every component of $G - N[P]$ has at most $n/2$ vertices.

We consider two cases, depending on the number of vertices of $P$.

\subsection*{Case A: $P$ has at most 6 vertices.}
Denote the consecutive vertices of $P$ by $x_1,\ldots,x_p$, where $p \leq 6$.
For $j \in [p]$, let $A_j$ be the set of vertices that are adjacent to $x_j$, but not to $x_{j'}$ for any $j' < j$.
Note that sets $A_1,\ldots,A_p$ form a partition of $N(P)$.
Let $A_{p+1}$ denote the set $V(G) \setminus N[P]$. We write $A_{> j}$ as a shorthand for $\bigcup_{j' >j}A_j$.

For each color $i \in [k]$ and each $j \in [p]$,
let $S_j^i$ be an inclusion-wise minimal subset of $X \cap A_j \cap c^{-1}(i)$
with the property that $N(X \cap A_j \cap c^{-1}(i)) \cap A_{>j} = N(S_j^i) \cap A_{>j}$.
As $x_j$ is complete to $A_j$ and anticomplete to $A_{>j}$, by \cref{lem:minimal} we observe that $|S_j^i|\leq 2$.

\paragraph{Separating the sets $A_j$ from each other.}

We exhaustively guess $X \cap P$, and $c|_{X \cap P}$, and sets $S_j^i$, for all $i \in [k]$ and $j \in [p]$.
This results in at most $(k+1)^6 \cdot n^{2pk} \leq n^{13k}$ branches (here we use that $n$ is large).
Note that we only consider branches that do not yield immediate contradiction.
In particular, for each $i \in [k]$, the set $\bigcup_{j \in [p]} S_j^i$ should be independent, should not have edges to vertices from $P \cap X$ colored $i$, and every $v \in \bigcup_{j \in [p]} S_j^i$ should satisfy $\rev(v,i) >0$.
Consider one such branch.

For each vertex $v \in P \cap X$, we forbid color $c(i)$ to any neighbor of $v$.
Similarly, for each vertex $v \in S^i_j$, we forbid color $i$ to any neighbor of $v$, and every color but $i$ to $v$.
Note that in the branch corresponding to correct guess this restriction is compatible the solution $(X,c)$.
Furthermore, any optimum solution of the current subinstance will contain every vertex from $S_j^i$, and its color will be $i$.
Now consider a vertex $v \in A_{>j}$ that is not adjacent to any vertex in $S_j^i$.
Note that this means that no neighbor of $v$ in $A_j$ is colored $i$ by $c$.
Thus, we forbid color $i$ to every vertex in $N(v) \cap A_j$.

\begin{claim}\label{clm:disjointlists}
    Let $uv \in E(G)$ such that $u \in A_j$ and $v \in A_{j'}$, where $j \neq j'$.
    The sets of colors with positive revenue for vertices $u$ and $v$ are disjoint.
\end{claim}

\begin{claimproof}
    By symmetry assume that $j < j'$. Consider a color $i \in [k]$.
    If $v$ is adjacent to a vertex from $S^i_j$ (possibly this vertex is $u$), then the color $i$ was forbidden to $v$, so $\rev(v,i) = 0$.
    Otherwise, if $v$ is anticomplete to  $S^i_j$, then the color $i$ was forbidden to every neighbor of $v$ in $A_j$. In particular, $\rev(u,i)=0$.    
\end{claimproof}

\paragraph{Solving the problem on each set $A_j$.}

As a consequence, the subinstances induced by sets $A_j$ can be solved independently, and their corresponding optimum solutions can be combined into an optimum solution for $(G,\rev)$ (including also vertices from $P \cap X$ that we guessed before).
Thus we call the algorithm recursively for subinstances: $(G[A_i],\rev)$ for $i \in [p]$, and $(C,\rev)$, for every component $C$ of $G - N[P]$.
Note that each instance of $(G[A_i],\rev)$ has clique number at most $\omega-1$, as a clique of size $\omega$ in $G[A_i]$, together with $x_i$, would give a larger clique in $G$.

In the $(\bull,\chair)$-free case, i.e., in the setting of \cref{thm:chair}, the same holds for instances $(C,\rev)$, as ensured by \cref{lem:structureCh}.
Unfortunately, in the  $(\bull,\E)$-free case, i.e., in the setting of \cref{thm:E}, the clique number in each component $C$ does not have to decrease. However, the instance becomes simpler for another reason: by \cref{thm:gyarfas} and \cref{lem:structureE}, we know that each component $C$ has at most $n/2$ vertices.

As mentioned before, the algorithm returns the union of the solutions found by these recursive calls, and the previously guessed set $P \cap X$ with its coloring.
In the branch where each guess is correct, i.e., corresponds to the structure of $(X,c)$, the returned solution will be at least as good as $(X,c)$. Consequently, the algorithm returns an optimum solution.

This concludes the description of Case A.

\subsection*{Case B: $P$ has at least 7 vertices.}
Now let us proceed to the case that $P$ has at least 7 vertices.
We call \cref{lem:structureE} to obtain the partition of $V(G)$ into sets $R,D,T$.

First, we exhaustively guess the set $A \subseteq [k]$ of colors that appear on vertices from $X \cap R$;
this results in $2^k$ branches. Consider one such branch.
We forbid all colors not in $A$ to all vertices in $R$.
Moreover, as $D$ is complete to $R$ and thus no vertex from $D$ is colored (by $c$) with a color from $A$,
we forbid all colors from $A$ to all vertices of $D$.
Note that this allows us to split the problem into two independent subinstances: $(G[R],\rev)$ and $(G - R,\rev)$, whose corresponding optimum solution can be combined into an optimum solution of $(G,\rev)$. 

\paragraph{Solving the problem on $R$.}
Let $V_1,\ldots,V_r$ be the partition sets of $R$.
As for every $i$, the set $V_i$ is complete to $V_{i+1}$ and each of them is non-empty,
we notice that the clique number of each $G[V_i]$ is at most $\omega-1$.
Indeed, otherwise a maximum clique in $G[V_i]$, together with any vertex from $V_{i-1}$ or $V_{i+1}$ (at least one of these sets exists) forms a clique of size $\omega+1$ in $G$, a contradiction.
Thus, for each $i \in [r]$, any subinstance with vertex set $V_i$ can be solved using the algorithm from inductive assumption.
Consequently, we can use \cref{lem:fatpath} to solve the subinstance $(G[R],\rev)$.

\paragraph{Separating $D$ and $T$.}
Now we would like to analyze the set subinstance $(G - R,\rev)$ and reduce it to a number of independent, simpler subinstances. We proceed similarly as in Case A:
For each color $i$, let $D^i$ be an inclusion-wise minimal subset of $D \cap X \cap c^{-1}(i)$ such that
$N(D^i) \cap T = N(D \cap X \cap c^{-1}(i)) \cap T$.
Let $r$ be any vertex from $R$. Since $r$ is complete to $D$ and anticomplete to $T$, by \cref{lem:minimal} we observe that $|D^i| \leq 2$.
Thus, we exhaustively guess all sets $D^i$, for $i \in [k]$.
This results in at most $n^{2k}$ branches.

Then, we modify the revenue function analogously as in Case A, so that it is compatible with the guessed sets.
An analogue of \cref{clm:disjointlists} applies here and the subinstances $(G[D],\rev)$ and $(C,\rev)$ for every component $C$ of $G - (R \cup D)$, can be solved independently.

\paragraph{Solving the problem on $D$ and $T$.}
We call the algorithm recursively for every such subinstance; let us explain why each of them is simpler than $(G,\rev)$. The arguments are analogous as in Case A.
The clique number of $G[D]$ is at most $\omega -1$, because any vertex from $R$ is complete to $D$.
In the $(\bull,\chair)$-free case, i.e., in the setting of \cref{thm:chair}, the same holds for each subinstance $(C,\rev)$, as ensured by \cref{lem:structureCh}.
In the $(\bull,\E)$-free case, i.e., in the setting of \cref{thm:E}, use the third statement of \cref{lem:structureE}.
As each component of $G - N[P]$ has at most $n/2$ vertices, the same holds for each component $C$ of $G - (R \cup D)$.

The optimum solutions found by each recursive calls, together with the optimum solution for $G[R],\rev)$, can be freely combined together to obtain a solution for $(G,\rev)$.
Similarly as in Case A, in the branch corresponding to the structure of $(X,c)$, the algorithm returns an optimum solution.

\subsection*{Complexity analysis.}

The complexity analysis is split into two cases: if $G$ is $(\bull,\chair)$-free (as in \cref{thm:chair}) or if $G$ is $(\bull,\E)$-free (as in \cref{thm:E}).

\paragraph{Complexity for $(\bull,\chair)$-free graphs.}
We aim to show that the complexity $F(n,\omega)$ of the algorithm is bounded by $n^{\tau \cdot k \cdot \omega}$, where $\tau$ is some absolute (but large) constant.

In case $A$, we have at most $n^{13k}$ directions, and in each of them we call the algorithm recursively on at most $n$ instances, each of clique number at most $\omega-1$. This yields the following recursion.
\[
F(n,\omega) \leq n^{13k} \cdot n \cdot F(n,\omega-1) \leq n^{13k+1} \cdot n^{\tau \cdot k \cdot (\omega-1)} \leq n^{\tau \cdot k \cdot \omega}
\]
where in the last inequality we use that $\tau$ is large.

In case $B$, we consider $2^k$ branches to choose $A$.
For each branch, we apply \cref{lem:fatpath} in order to solve the subinstance $(G[R],\rev)$ in time $\Oh(6^k \cdot n \cdot F(n,\omega-1))$.
This gives us
\begin{align*}
F(n,\omega) \leq & \ 2^k \left(\Oh(6^k \cdot n \cdot F(n,\omega-1)) +  n^{2k} \cdot n \cdot F(n,\omega-1) \right) \\ \leq & \ n^{3k+1} \cdot n^{\tau \cdot k \cdot (\omega-1)} \leq n^{\tau \cdot k \cdot \omega}.
\end{align*}
This completes the proof of \cref{thm:chair}.

\paragraph{Complexity for $(\bull,\E)$-free graphs.}

Now we want to show that the complexity $F(n,\omega)$ of the algorithm is bounded by $n^{\tau \cdot k \cdot \omega \cdot \log n}$, where $\tau$ is some absolute (but large) constant.
The complexity analysis is analogous as for $(\bull,\chair)$-free graphs, with the difference that now in some recursive calls the clique number does not drop. However, in these calls the number of vertices decreases by half.
Thus, by the inductive assumption, the complexity of each recursive call is upper-bounded by 
\begin{align*}
\max \left ( F(n/2, \omega),  F(n,\omega-1) \right) \leq & \ \max \left ( n^{\tau k \cdot (\omega-1) \cdot \log n}, n^{\tau k \cdot \omega \cdot (\log n -1) }  \right) \\ \leq & \ n^{\tau k \cdot \omega \cdot \log n - \tau \cdot k}.
\end{align*}

Thus, in Case A we obtain:
\[
F(n,\omega) \leq n^{13k+1} \cdot  n^{\tau \cdot k \cdot \omega \cdot \log n - \tau \cdot  k} \leq n^{\tau \cdot k \cdot \omega \cdot \log n},
\]
and in Case B we obtain:
\[
F(n,\omega)  \leq n^{3k+1} \cdot n^{\tau \cdot k \log n - \tau k} \leq n^{\tau \cdot k \cdot \omega \cdot \log n}.
\]
This completes the proof of \cref{thm:E}.

\section{Algorithmic corollaries}
Recall that \textsc{List $k$-Coloring} can be seen as a special case of \textsc{Max Partial $k$-Coloring}. However, when solving \textsc{List $k$-Coloring}, we can safely reject any graph with clique number larger than $k$.
Thus, we may assume that $\omega \leq k$. This observation, combined with \cref{thm:chair,thm:E}, immediately yields \cref{cor:chair,cor:E} which we restate here.

\corchair*
\corE*

In order to prove \cref{cor:subexp}, we use the win-win approach of Chudnovsky et al.~\cite{DBLP:journals/siamdm/ChudnovskyKPRS21}.

\corsubexp*
\begin{proof}
    Let $(G,\rev)$ the the instance with $n$ vertices.
    We exhaustively check if $G$ contains a clique with $\sqrt{n/\log n}$ vertices;
    this can be done in time $n^{\Oh(\sqrt{n/\log n})} = 2^{\Oh(\sqrt{n \log n})}$.

    First, suppose that such a clique exists, call it $K$.
    Note that at most $k$ vertices from $K$ may be chosen to any solution.
    We branch into all possible ways of choosing these vertices and assigning their colors. Note that this gives at most $\sqrt{n / \log n}^k$ branches.

    For each branch, we modify the revenue function by forbidding the color $i$ to every neighbor of a vertex that was chosen to be colored $i$.
    Next, we remove $K$ from the graph and proceed recursively.

    The complexity in this case is given by the recursion 
    \[
        F(n) \leq \sqrt{n / \log n}^k F(n - \sqrt{(n/ \log n)}) = n^{\Oh(nk \sqrt{ \log n} / \sqrt{n})} = 2^{\Oh(k \cdot \sqrt{n} \log^{3/2} n)}.
    \]

    In the other case, we call the algorithm from \cref{thm:E} with $\omega = \sqrt{n/\log n}$. The complexity in this case is thus
    \[
    F(n) \leq n^{\Oh(k \cdot \sqrt{n/\log n} \log n)}= 2^{\Oh(k \cdot \sqrt{n} \log^{3/2} n)}.
    \]
    This completes the proof.    
\end{proof}

\section{Conclusion}

As already discussed in the introduction, the problem considered in this paper, i.e., \textsc{Max Partial $k$-Coloring}, was already considered in the setting of restricted graph classes~\cite{DBLP:journals/siamdm/ChudnovskyKPRS21,DBLP:journals/corr/abs-2410-21569,henderson2024maximumkcolourableinducedsubgraphs}.
Actually, the problem appearing in (some of ) the prior work is even more general.
For graphs $G$ and $H$, an \emph{$H$-coloring} of $G$ is a function $c : V(G) \to V(H)$ such that $c(u)c(v) \in E(H)$ for every $uv \in E(G)$. Note that if $H$ is the complete graph on $k$ vertices, then $H$-colorings coincide with proper $k$-colorings. Thus $H$-colorings can be seen as a far-reaching generalization of $k$-colorings.
The analogue of \textsc{Max Partial $k$-Coloring}, but for $H$-colorings, if defined as follows.
Here $H$ is a fixed simple graph.

\problemTask{\textsc{Max Partial $H$-Coloring}}
{a graph $G$, a revenue function $\rev : V(G) \times V(H) \to \mathbb{Q}_{\geq 0}$}
{a set $X$ and an $H$-coloring $c$ of $G[X]$, such that \[\sum_{v \in X} \rev(v,c(v))\] is maximum possible }

It is straightforward to verify that the proof of \cref{thm:chair,thm:E} and subsequent corollaries translates directly to \textsc{Max Partial $H$-Coloring}. In particular, we can obtain the following two stronger statements.
(We decided to keep the paper focused on (partial) $k$-colorings for simplicity of exposition, as the generalization to $H$-colorings does not bring any new insight.)

\begin{theorem}\label{thm:Hchair}
    For every simple graph $H$, \textsc{Max Partial $H$-Coloring} on $(\bull,\chair)$-free instances with $n$ vertices and clique number $\omega$ can be solved in time $n^{\Oh(|V(H)| \omega)}$.
\end{theorem}

\begin{theorem}\label{thm:HE}
    For every simple graph $H$, \textsc{Max Partial $H$-Coloring} on $(\bull,\E)$-free instances with $n$ vertices and clique number $\omega$ can be solved in time $n^{\Oh(|V(H)| \omega  \log n)}$.
\end{theorem}

Let us emphasize that in \cref{thm:Hchair,thm:HE} we assume that $H$ is a simple graph, so in particular its vertices have no loops. However, note that in the context of \textsc{Max Partial $H$-Coloring}, it makes sense to consider graphs $H$ with loops on vertices: even though mapping everything to a vertex with a loop is a valid $H$-coloring, it might be far from optimal due to the values of the revenue function.
However, there is no hope to strengthen \cref{thm:Hchair,thm:HE} to all graphs $H$ with possible loops.
Indeed, Chudnovsky et al.~\cite{DBLP:journals/siamdm/ChudnovskyKPRS21} showed an example of an 12-vertex graph $H_0$ with loops on all vertices, such that \textsc{Max Partial $H_0$-Coloring} (and even more restrictive \textsc{List $H_0$-Coloring}, i.e., the analogue of \textsc{List $k$-Coloring} for $H_0$-colorings) is \textsf{NP}-hard in complements of bipartite graphs. Observe that complements of bipartite graphs are in particular $(\bull,\chair)$-free and thus also $(\bull,\E)$-free. Furthermore, the hardness reduction also excludes any algorithm working in time $2^{o(n)}$, assuming the Exponential-Time Hypothesis (ETH).

Nevertheless, we believe that even the following stronger variant of \cref{conj} holds.

\begin{conjecture}[Strengthening of \cref{conj}]
    For every simple graph $H$, \textsc{Max Partial $H$-Coloring} on $(\bull,\E)$-free graphs can be solved in polynomial time.
\end{conjecture}

\bibliographystyle{abbrv}
\bibliography{main}
\end{document}

%% file: header.tex
\usepackage[T1]{fontenc}
\usepackage[margin=2cm]{geometry}
\usepackage{libertine}
\usepackage{amsfonts}
\usepackage{amsmath, amsthm}
\usepackage{amssymb}
\usepackage{todonotes}
\usepackage[capitalise]{cleveref}
\usepackage{authblk}
\usepackage{thm-restate}
\usepackage{enumerate}
\usepackage[noadjust]{cite}

\newtheorem{theorem}{Theorem}[section]
\newtheorem{lemma}[theorem]{Lemma}
\newtheorem{claim}{Claim}[theorem]

\newtheorem{conjecture}[theorem]{Conjecture}

\newenvironment{claimproof}{\noindent{\emph{Proof of Claim.}}}{\hfill$\triangleleft$\medskip}

\crefname{claim}{Claim}{Claims}
\crefname{lemma}{Lemma}{Lemmas}

\newcommand{\bull}{\textrm{bull}}
\newcommand{\E}{\textsf{E}}
\newcommand{\chair}{\textrm{chair}}
\newcommand{\rev}{\mathsf{rev}}
\newcommand{\Oh}{\mathcal{O}}
\newcommand{\cF}{\mathcal{F}}

\usepackage{framed}
\usepackage{tabularx}
 \usepackage{tikz}
 \usetikzlibrary{calc} 

\newlength{\RoundedBoxWidth}
\newsavebox{\GrayRoundedBox}
\newenvironment{GrayBox}[1]%
   {\setlength{\RoundedBoxWidth}{.93\columnwidth}
    \def\boxheading{#1}
    \begin{lrbox}{\GrayRoundedBox}
       \begin{minipage}{\RoundedBoxWidth}}%
   {   \end{minipage}
    \end{lrbox}
    \begin{center}
    \begin{tikzpicture}%
       \node(Text)[draw=black!20,fill=white,rounded corners,inner sep=2ex,text width=\RoundedBoxWidth]
             {\usebox{\GrayRoundedBox}};
        \coordinate(x) at (current bounding box.north west);
        \node [draw=white,rectangle,inner sep=3pt,anchor=north west,fill=white]
        at ($(x)+(6pt,.75em)$) {\boxheading};
    \end{tikzpicture}
    \end{center}}

\newenvironment{defproblemx}[1]{\noindent\ignorespaces%
                                \FrameSep=6pt%
                                \parindent=0pt%
                \begin{GrayBox}{#1}%
                \begin{tabular*}{\columnwidth}{!{\extracolsep{\fill}}@{\hspace{.1em}} >{\itshape} p{1.5cm} p{0.86\columnwidth} @{}}%
            }{
                \end{tabular*}%
                \end{GrayBox}%
                \ignorespacesafterend
            }

\newcommand{\problemTask}[3]{%
  \begin{defproblemx}{#1}
    Input: & #2 \\
    Task: & #3
  \end{defproblemx}
}

%% file: main.bbl
\begin{thebibliography}{10}

\bibitem{DBLP:journals/siamcomp/AbrishamiCPRS24}
T.~Abrishami, M.~Chudnovsky, M.~Pilipczuk, P.~Rzążewski, and P.~D. Seymour.
\newblock Induced subgraphs of bounded treewidth and the container method.
\newblock {\em {SIAM} J. Comput.}, 53(3):624--647, 2024.

\bibitem{10.1145/3708544}
A.~Agrawal, P.~T. Lima, D.~Lokshtanov, P.~Rz\k{a}\.{z}ewski, S.~Saurabh, and R.~Sharma.
\newblock Odd cycle transversal on {$P_5$}-free graphs in polynomial time.
\newblock {\em ACM Trans. Algorithms}, 21(2), Jan. 2025.

\bibitem{DBLP:journals/algorithmica/BacsoLMPTL19}
G.~Bacs{\'{o}}, D.~Lokshtanov, D.~Marx, M.~Pilipczuk, Z.~Tuza, and E.~J. van Leeuwen.
\newblock Subexponential-time algorithms for maximum independent set in {$P_t$}-free and broom-free graphs.
\newblock {\em Algorithmica}, 81(2):421--438, 2019.

\bibitem{DBLP:journals/combinatorica/BonomoCMSSZ18}
F.~Bonomo, M.~Chudnovsky, P.~Maceli, O.~Schaudt, M.~Stein, and M.~Zhong.
\newblock Three-coloring and list three-coloring of graphs without induced paths on seven vertices.
\newblock {\em Comb.}, 38(4):779--801, 2018.

\bibitem{DBLP:journals/tcs/Bonomo-Braberman21}
F.~Bonomo{-}Braberman, M.~Chudnovsky, J.~Goedgebeur, P.~Maceli, O.~Schaudt, M.~Stein, and M.~Zhong.
\newblock Better 3-coloring algorithms: Excluding a triangle and a seven vertex path.
\newblock {\em Theor. Comput. Sci.}, 850:98--115, 2021.

\bibitem{DBLP:journals/endm/BrauseSHRVK15}
C.~Brause, I.~Schiermeyer, P.~Holub, Z.~Ryj{\'{a}}\v{c}ek, P.~Vr{\'{a}}na, and R.~Krivos{-}Bellus.
\newblock 4-colorability of {$P_6$}-free graphs.
\newblock {\em Electron. Notes Discret. Math.}, 49:37--42, 2015.

\bibitem{DBLP:journals/tcs/BroersmaGPS12}
H.~Broersma, P.~A. Golovach, D.~Paulusma, and J.~Song.
\newblock Updating the complexity status of coloring graphs without a fixed induced linear forest.
\newblock {\em Theor. Comput. Sci.}, 414(1):9--19, 2012.

\bibitem{DBLP:journals/corr/abs-2412-14836}
M.~Chudnovsky, J.~Czyzewska, K.~Kluk, M.~Pilipczuk, and P.~Rzazewski.
\newblock Sparse induced subgraphs in {$P_7$}-free graphs of bounded clique number.
\newblock {\em CoRR}, abs/2412.14836, 2024.

\bibitem{DBLP:journals/siamdm/ChudnovskyKPRS21}
M.~Chudnovsky, J.~King, M.~Pilipczuk, P.~Rz\k{a}\.zewski, and S.~Spirkl.
\newblock Finding large {$H$}-colorable subgraphs in hereditary graph classes.
\newblock {\em {SIAM} J. Discret. Math.}, 35(4):2357--2386, 2021.

\bibitem{DBLP:journals/jgt/ChudnovskyMSZ17}
M.~Chudnovsky, P.~Maceli, J.~Stacho, and M.~Zhong.
\newblock 4-coloring {$P_6$}-free graphs with no induced 5-cycles.
\newblock {\em J. Graph Theory}, 84(3):262--285, 2017.

\bibitem{DBLP:journals/dagstuhl-reports/ChudnovskyPS19}
M.~Chudnovsky, D.~Paulusma, and O.~Schaudt.
\newblock {Graph Colouring: from Structure to Algorithms (Dagstuhl Seminar 19271)}.
\newblock {\em Dagstuhl Reports}, 9(6):125--142, 2019.

\bibitem{DBLP:journals/siamcomp/ChudnovskyPPT24}
M.~Chudnovsky, M.~Pilipczuk, M.~Pilipczuk, and S.~Thomass{\'{e}}.
\newblock Quasi-polynomial time approximation schemes for the maximum weight independent set problem in {$H$}-free graphs.
\newblock {\em {SIAM} J. Comput.}, 53(1):47--86, 2024.

\bibitem{DBLP:journals/dm/ChudnovskySZ20}
M.~Chudnovsky, S.~Spirkl, and M.~Zhong.
\newblock List 3-coloring {$P_t$}-free graphs with no induced 1-subdivision of {$K_{1,s}$}.
\newblock {\em Discret. Math.}, 343(11):112086, 2020.

\bibitem{DBLP:journals/siamcomp/ChudnovskySZ24}
M.~Chudnovsky, S.~Spirkl, and M.~Zhong.
\newblock Four-coloring {$P_6$}-free graphs. {I. Extending} an excellent precoloring.
\newblock {\em {SIAM} J. Comput.}, 53(1):111--145, 2024.

\bibitem{DBLP:journals/siamcomp/ChudnovskySZ24a}
M.~Chudnovsky, S.~Spirkl, and M.~Zhong.
\newblock Four-coloring {$P_6$}-free graphs. {II. Finding} an excellent precoloring.
\newblock {\em {SIAM} J. Comput.}, 53(1):146--187, 2024.

\bibitem{DBLP:journals/algorithmica/DabrowskiFJPPR20}
K.~K. Dabrowski, C.~Feghali, M.~Johnson, G.~Paesani, D.~Paulusma, and P.~Rzążewski.
\newblock On cycle transversals and their connected variants in the absence of a small linear forest.
\newblock {\em Algorithmica}, 82(10):2841--2866, 2020.

\bibitem{DBLP:journals/ipl/DabrowskiP18}
K.~K. Dabrowski and D.~Paulusma.
\newblock On colouring {$(2P_2,H)$}-free and {$(P_5,H)$}-free graphs.
\newblock {\em Inf. Process. Lett.}, 134:35--41, 2018.

\bibitem{DBLP:journals/cpc/Emden-WeinertHK98}
T.~Emden{-}Weinert, S.~Hougardy, and B.~Kreuter.
\newblock Uniquely colourable graphs and the hardness of colouring graphs of large girth.
\newblock {\em Comb. Probab. Comput.}, 7(4):375--386, 1998.

\bibitem{DBLP:conf/focs/GartlandL20}
P.~Gartland and D.~Lokshtanov.
\newblock Independent set on {$P_k$}-free graphs in quasi-polynomial time.
\newblock In S.~Irani, editor, {\em 61st {IEEE} Annual Symposium on Foundations of Computer Science, {FOCS} 2020, Durham, NC, USA, November 16-19, 2020}, pages 613--624. {IEEE}, 2020.

\bibitem{DBLP:conf/stoc/GartlandLMPPR24}
P.~Gartland, D.~Lokshtanov, T.~Masa\v{r}{\'{\i}}k, M.~Pilipczuk, M.~Pilipczuk, and P.~Rzążewski.
\newblock Maximum weight independent set in graphs with no long claws in quasi-polynomial time.
\newblock In B.~Mohar, I.~Shinkar, and R.~O'Donnell, editors, {\em Proceedings of the 56th Annual {ACM} Symposium on Theory of Computing, {STOC} 2024, Vancouver, BC, Canada, June 24-28, 2024}, pages 683--691. {ACM}, 2024.

\bibitem{DBLP:journals/jcss/GaspersHP19}
S.~Gaspers, S.~Huang, and D.~Paulusma.
\newblock Colouring square-free graphs without long induced paths.
\newblock {\em J. Comput. Syst. Sci.}, 106:60--79, 2019.

\bibitem{DBLP:journals/jgt/GolovachJPS17}
P.~A. Golovach, M.~Johnson, D.~Paulusma, and J.~Song.
\newblock A survey on the computational complexity of coloring graphs with forbidden subgraphs.
\newblock {\em J. Graph Theory}, 84(4):331--363, 2017.

\bibitem{DBLP:journals/iandc/GolovachPS14}
P.~A. Golovach, D.~Paulusma, and J.~Song.
\newblock Closing complexity gaps for coloring problems on {$H$}-free graphs.
\newblock {\em Inf. Comput.}, 237:204--214, 2014.

\bibitem{DBLP:journals/talg/GrzesikKPP22}
A.~Grzesik, T.~Klimo\v{s}ov{\'{a}}, M.~Pilipczuk, and M.~Pilipczuk.
\newblock Polynomial-time algorithm for maximum weight independent set on {$P_6$}-free graphs.
\newblock {\em {ACM} Trans. Algorithms}, 18(1):4:1--4:57, 2022.

\bibitem{gyarfas2}
A.~Gy\'{a}rf\'{a}s.
\newblock Problems from the world surrounding perfect graphs.
\newblock In {\em Proceedings of the International Conference on Combinatorial Analysis and its Applications, (Pokrzywna, 1985)}, number~19 in Zastos. Mat., pages 413--441, 1987.

\bibitem{henderson2024maximumkcolourableinducedsubgraphs}
C.~Henderson, E.~Smith-Roberge, S.~Spirkl, and R.~Whitman.
\newblock Maximum $k$-colourable induced subgraphs in {$(P_5+rK_1)$}-free graphs.
\newblock {\em \rm Preprint available at \url{https://arxiv.org/abs/2410.08077}}, 2024.

\bibitem{DBLP:journals/algorithmica/HoangKLSS10}
C.~T. Ho{\`{a}}ng, M.~Kami\'nski, V.~V. Lozin, J.~Sawada, and X.~Shu.
\newblock Deciding \emph{k}-colorability of {$P_5$}-free graphs in polynomial time.
\newblock {\em Algorithmica}, 57(1):74--81, 2010.

\bibitem{hodur20243colourabilitybullhfreegraphs}
N.~Hodur, M.~Pilśniak, M.~Prorok, and I.~Schiermeyer.
\newblock On 3-colourability of {$(bull, H)$}-free graphs.
\newblock {\em \rm Preprint available at \url{https://arxiv.org/abs/2404.12515}}, 2024.

\bibitem{DBLP:journals/siamcomp/Holyer81a}
I.~Holyer.
\newblock The {NP}-completeness of edge-coloring.
\newblock {\em {SIAM} J. Comput.}, 10(4):718--720, 1981.

\bibitem{DBLP:journals/ejc/Huang16}
S.~Huang.
\newblock Improved complexity results on k-coloring {$P_t$}-free graphs.
\newblock {\em Eur. J. Comb.}, 51:336--346, 2016.

\bibitem{DBLP:journals/algorithmica/JelinekKMNP22}
V.~Jel{\'{\i}}nek, T.~Klimo\v{s}ov{\'{a}}, T.~Masa\v{r}{\'{\i}}k, J.~Novotn{\'{a}}, and A.~Pokorn{\'{a}}.
\newblock On 3-coloring of {$(2P_4,C_5$)}-free graphs.
\newblock {\em Algorithmica}, 84(6):1526--1547, 2022.

\bibitem{DBLP:journals/jal/LevenG83}
D.~Leven and Z.~Galil.
\newblock {NP} completeness of finding the chromatic index of regular graphs.
\newblock {\em J. Algorithms}, 4(1):35--44, 1983.

\bibitem{DBLP:journals/corr/abs-2410-21569}
D.~Lokshtanov, P.~Rzazewski, S.~Saurabh, R.~Sharma, and M.~Zehavi.
\newblock Maximum partial list {$H$}-coloring on {$P_5$}-free graphs in polynomial time.
\newblock {\em CoRR}, abs/2410.21569, 2024.

\bibitem{DBLP:conf/soda/LokshantovVV14}
D.~Lokshtanov, M.~Vatshelle, and Y.~Villanger.
\newblock Independent set in {$P_5$}-free graphs in polynomial time.
\newblock In C.~Chekuri, editor, {\em Proceedings of the Twenty-Fifth Annual {ACM-SIAM} Symposium on Discrete Algorithms, {SODA} 2014, Portland, Oregon, USA, January 5-7, 2014}, pages 570--581. {SIAM}, 2014.

\bibitem{DBLP:conf/sosa/PilipczukPR21}
M.~Pilipczuk, M.~Pilipczuk, and P.~Rzazewski.
\newblock Quasi-polynomial-time algorithm for independent set in {$P_t$}-free graphs via shrinking the space of induced paths.
\newblock In H.~V. Le and V.~King, editors, {\em 4th Symposium on Simplicity in Algorithms, {SOSA} 2021, Virtual Conference, January 11-12, 2021}, pages 204--209. {SIAM}, 2021.

\bibitem{DBLP:journals/dam/RanderathS04}
B.~Randerath and I.~Schiermeyer.
\newblock 3-colorability in {P} for {$P_6$}-free graphs.
\newblock {\em Discret. Appl. Math.}, 136(2-3):299--313, 2004.

\bibitem{DBLP:journals/gc/RanderathS04}
B.~Randerath and I.~Schiermeyer.
\newblock Vertex colouring and forbidden subgraphs - {A} survey.
\newblock {\em Graphs Comb.}, 20(1):1--40, 2004.

\bibitem{DBLP:journals/sigact/Stockmeyer73}
L.~J. Stockmeyer.
\newblock Planar 3-colorability is polynomial complete.
\newblock {\em {SIGACT} News}, 5(3):19--25, 1973.

\bibitem{DBLP:conf/icalp/TedderCHP08}
M.~Tedder, D.~G. Corneil, M.~Habib, and C.~Paul.
\newblock Simpler linear-time modular decomposition via recursive factorizing permutations.
\newblock In L.~Aceto, I.~Damg{\aa}rd, L.~A. Goldberg, M.~M. Halld{\'{o}}rsson, A.~Ing{\'{o}}lfsd{\'{o}}ttir, and I.~Walukiewicz, editors, {\em Automata, Languages and Programming, 35th International Colloquium, {ICALP} 2008, Reykjavik, Iceland, July 7-11, 2008, Proceedings, Part {I:} Tack {A:} Algorithms, Automata, Complexity, and Games}, volume 5125 of {\em Lecture Notes in Computer Science}, pages 634--645. Springer, 2008.

\end{thebibliography}
